\documentclass[12pt, draftclsnofoot, onecolumn]{IEEEtran}

\usepackage[T1]{fontenc}
\usepackage{graphicx}
\usepackage{cite}
\usepackage{subcaption}
\usepackage{overpic}
\usepackage{amssymb}
\usepackage[cmex10]{amsmath}
\usepackage{amsthm}
\usepackage{url}
\usepackage{bbm}
\usepackage{color}

\IEEEoverridecommandlockouts

\def\CN{\mathcal{N}_{\mathbb{C}}}
\def\Real{\mathbb{R}}
\def\Complex{\mathbb{C}}
\def\Integer{\mathbb{Z}}

\def\Ex{\mathbb{E}}
\def\sinc{\mathrm{sinc}}
 
\def\imagunit{i} 

\theoremstyle{plain}

\newtheorem{lemma}{Lemma}
\newtheorem{corollary}{Corollary}

\newtheorem{definition}{Definition}

\newcommand{\vect}[1]{{\bf{#1}}}

\graphicspath{{./images/}}

\usepackage[margin=0.90in]{geometry}

\begin{document}

\title{\huge{Spatially-Stationary Model for \\Holographic MIMO Small-Scale Fading}\vspace{-0.3cm}}

\author{
\IEEEauthorblockN{Andrea Pizzo, \emph{Member, IEEE}, Thomas L. Marzetta, \emph{Fellow, IEEE}, Luca Sanguinetti, \emph{Senior Member, IEEE}
\thanks{\vspace{-0.5cm}
\newline \indent Part of this work was presented at ISIT18 \cite{Marzetta2018}. A. Pizzo and T. Marzetta are with the Department of Electrical and Computer Engineering, Tandon School of Engineering, 11201 Brooklyn, NY (andrea.pizzo@nyu.edu, tom.marzetta@nyu.edu). L.~Sanguinetti is with the Dipartimento di Ingegneria dell'Informazione, University of Pisa, 56122 Pisa, Italy (luca.sanguinetti@unipi.it).}
}}
\maketitle

 \vspace{-1.5cm}
 
\begin{abstract} \vspace{-0.2cm}
Imagine an array with a massive (possibly uncountably infinite) number of antennas in a compact space. We refer to a system of this sort as Holographic MIMO.
 Given the impressive properties of Massive MIMO, one might expect a holographic array to realize extreme spatial resolution, incredible energy efficiency, and unprecedented spectral efficiency. At present, however, {its fundamental limits have not been conclusively established.} A major challenge for the analysis and understanding of such a paradigm shift is the lack of mathematically tractable and numerically reproducible channel models that retain some semblance to the physical reality. Detailed physical models are, in general, too complex for tractable analysis.
This paper aims to take a closer look at this interdisciplinary challenge. Particularly, we consider the small-scale fading in the far-field, and we model it as a zero-mean, spatially-stationary, and correlated Gaussian scalar random field. A physically-meaningful correlation is obtained by requiring that the random field be consistent with the scalar Helmholtz equation. This formulation leads directly to a rather simple and exact description of the three-dimensional small-scale fading as a Fourier plane-wave spectral representation.
Suitably discretized, this yields a discrete representation for the field as a Fourier plane-wave series expansion, from which a computationally efficient way to generate samples of the small-scale fading over spatially-constrained compact spaces is developed. The connections with the conventional tools of linear systems theory and Fourier transform are thoroughly discussed.
\vspace{-0.0cm}
\end{abstract}


\vspace{-0.5cm}
\section{Introduction} \label{sec:Introduction}
Massive MIMO refers to a wireless network technology where the base stations are equipped with a very large number $N$ of antennas to serve a multitude of $K$ terminals by spatial multiplexing \cite{Marzetta2010}. Thanks to the intense research performed over the last decade, 
Massive MIMO is today a mature technology \cite{MarzettaBook,LucaBook} whose key ingredients have made it into the 5G New Radio standard \cite{Five}. Its advantages in terms of spectral efficiency \cite{Ngo2013,UnlimitedCapacity}, energy efficiency \cite{Sanguinetti_2015}, and power control \cite{Cheng2017} are well understood and recognized. The channel capacity was shown to increase theoretically unboundedly in the regime where $N\to \infty$ while $K$ is fixed \cite{UnlimitedCapacity}. In practice, however, the number of antennas that fits into the common form factor of a base station site is fundamentally limited.
Hence, a natural question is \cite{Sanguinetti2019a}: \emph{how can we {practically} approach the limit $N\to \infty$?} 
{One solution is to integrate a massive (possibly infinite) number of antennas into a {compact} space, that is, a Holographic MIMO array.
In its asymptotic form, this can be thought of as a spatially-continuous electromagnetic aperture having an uncountably infinite number of antennas separated by an infinitesimal distance. This is the ultimate form of a Holographic MIMO array as $N\to \infty$.
Research in this field is taking place under the names of holographic radio-frequency systems \cite{Prather2017}, and large intelligent surfaces \cite{Rusek2018}. 
The technology developed in \cite{Pivotal}, known as holographic beamforming, is the first step in this direction.
 \vspace{-0.3cm}
\subsection{Motivation}
 {Realistic performance assessment of Holographic MIMO technologies requires the use of a channel model that reflects the main characteristics of a massive number of antennas in a compact space. The wireless channel is typically composed of the so-called large-scale fading and small-scale fading \cite{TseBook}. The former occurs on a larger scale -- typically a few hundred wavelengths -- and is due to distance-dependent pathloss, shadowing, antenna gains, and penetration losses, while the latter is a microscopic effect caused by small variations in the propagation environment. Both play a key role in wireless communications. This paper only considers the small-scale fading. In the multiple antenna literature, it is described by a complex-valued vector ${\bf h} \in {\mathbb{C}^N}$, which is modeled either deterministically or stochastically (e.g., \cite[Sec. 7.3]{LucaBook}). Examples of deterministic models are ray-tracing models and recorded channel measurements. The perfect line-of-sight (LoS) propagation model is another one, which is used in \cite{Rusek2018,Dardari2019a} to evaluate the capacity and compute the degrees of freedom of a planar continuous electromagnetic surface (referred to as a large intelligent surface). A common drawback of deterministic models is that they are only valid for specific scenarios. 
 
Unlike deterministic models, stochastic approaches are independent of a particular environment and, consequently, allow for far-reaching general conclusions. A well-known stochastic model in the far-field with no line-of-sight (NLoS) path is the independent and identically distributed (i.i.d.) Rayleigh fading, where $\vect{h} \sim \CN( \vect{0}_N, \vect{I}_N)$ is modeled as a circularly symmetric complex Gaussian vector with zero-mean and covariance matrix $\vect{I}_N$. 
This model has been (and still is) the basis of most theoretical research in multiple antenna systems. For example, in Massive MIMO it leads to neat, understandable closed-form spectral efficiency expressions \cite{MarzettaBook}. However, it is widely recognized to be inadequate when antenna spacing reduces, and spatial correlation naturally arises \cite{Sanguinetti2019a}. A tractable way to model spatially-correlated channels with NLoS path is the correlated Rayleigh fading model \cite{Sanguinetti2019a}: $\vect{h} \sim \CN( \vect{0}_N, \vect{R})$, 
where $\vect{R} \in \mathbb{C}^{N \times N}$ is the spatial correlation matrix. Different models exist for the generation of $\vect{R}$. Consistency with the physics laws of propagation implies that such models are based on a superposition of plane-waves (e.g., \cite[Sec. 7.3]{LucaBook}), which shifts the focus on the characterization of their amplitudes only.
The vast majority is driven by analytical arguments, which allow us to capture key characteristics of the propagation environment but inevitably tend to leave out significant physical phenomena. Moreover, the spatial correlation matrix is typically modeled to capture the propagation environment and array geometry jointly. 
For example, the i.i.d. Rayleigh fading models the small-scale fading over a half-wavelength spaced linear array surrounded by an isotropic propagation environment. 
Their combined action makes it hard to infer the physical properties of the propagation environment as antenna spacing reduces. For example, arrays with larger antenna spacing in a less scattered environment yield the same spatial correlation as arrays with smaller antenna spacing but in a more scattered environment \cite{Poon2004,Poon2006}. 


A stochastic model for the small-scale fading that retains some semblance to the physical reality is the widely known Clarke's model, e.g., \cite{Clarke} and \cite{Aulin79} for a two-dimensional (2D) and three-dimensional (3D) analysis, respectively. The model assumes a NLoS scenario with scalar radio waves (i.e., no polarization) propagating in the far-field of an isotropic (i.e., no dominant spatial directivity) random scattering environment \cite{Aris2000}. Under these circumstances, Clarke's model is exact, and its autocorrelation function between two arbitrary points in 2D and 3D spaces at a distance $r$ is, respectively, equal to $J_0(2 \pi r/\lambda)$ and $\sinc(2 r/\lambda)$ with $\lambda$ being the wavelength \cite{PaulrajBook,MolischBook}.\footnote{We use the definition $\sinc(x)=\sin(\pi x)/\pi x$ throughout this paper.} Notably, the channel exhibits spatial correlation even if there is no spatial directivity.

  
 This paper aims to generalize Clarke's model to \emph{non-isotropic} random scattering environments, with the purpose to obtain a generalized stochastic model for the small-scale fading of 3D spatially-stationary channels, which is: $(i)$ physically meaningful; $(ii)$ mathematically tractable; and $(iii)$ numerically reproducible. 
By focusing on a continuous formulation of the problem, rather than discrete, we simplify the theoretical analysis and obtain a rather simple expression for the small-scale fading. Also, spatial sampling and discrete formulation somehow tend to hide fundamental results, which are otherwise revealed by a continuous analysis \cite{Baggeroer}.
Previous works in this direction can be found in \cite{Poon2004, Poon2006}, and \cite{Debbah2015}. Unlike these works, we treat radio wave propagation as a linear system, which provides us with a simple and intuitive interpretation of the proposed model through linear system theory and Fourier transform, without the recourse to special functions (i.e., Green function \cite{Poon2004, Poon2006}, spherical harmonics \cite{Debbah2015}), detailed parametric models.

 \vspace{-0.2cm}
\subsection{Contribution}
We consider the 3D small-scale fading in the far-field
and assume that it can be modeled as a zero-mean, spatially-stationary, and correlated Gaussian scalar random field satisfying the Helmholtz equation in the frequency domain (equivalent to the scalar wave equation in the time domain), as dictated by physics. This modeling yields the \emph{only physically-meaningful} spatial correlation function, whose power spectral density (in the spatial-frequency or \emph{wavenumber} domain) is impulsive with support on the surface of a sphere of radius ${\kappa = 2\pi/\lambda}$, and uniquely described by a \emph{spectral factor} that specifies directionality and physically characterizes the propagation environment in its most general form. The structure of the spatial correlation function leads directly to the so-called 2D Fourier plane-wave spectral representation for the field, which is given by a superposition of a continuum of plane-waves having zero-mean statistically-independent Gaussian-distributed random amplitudes. 
We discuss the connection between the derived formula and the space-wavenumber Fourier spectral representation, which represents the spatial counterpart of the time-frequency mapping for time-domain random processes. Notably, we show that the small-scale fading has a singularly-integrable bandlimited spectrum in the wavenumber domain that is defined on a disk of radius $\kappa$. This is a direct consequence of the Helmholtz equation, which acts as a 2D linear space-invariant physical filter that projects the number of observable field configurations onto a lower-dimensional space.
Comparisons with Clarke's model are also made to show that this is the closest physically-tenable model to i.i.d. Rayleigh fading.

Also, the bandlimited nature of the 2D Fourier plane-wave spectral representation is exploited to statistically characterize the small-scale fading over Holographic MIMO arrays of compact size.
The result is a 2D Fourier plane-wave series expansion having a countably-finite number of zero-mean statistically-independent Gaussian-distributed random coefficients with given finite variance. We show that the provided model can be interpreted as the asymptotic version of the famous Karhunen-Lo{\`e}ve series expansion.
Finally, spatial discretization of the field yields a numerical procedure to generate small-scale fading samples through inverse fast Fourier transform efficiently.

\emph{Reproducible Research}: The Matlab code package to reproduce all numerical results is available at \url{https://github.com/lucasanguinetti/Holographic-MIMO-Small-Scale-Fading}.

 \vspace{-0.3cm}
\subsection{Outline and Notation}
The rest of this paper is organized as follows. Section \ref{sec:Model} introduces the continuous physical modeling for the 3D small-scale fading and studies the implication of the Helmholtz equation on its second-order statistics. Section \ref{sec:Fourier_spectral} derives the 2D Fourier plane-wave spectral representation and elucidates the connection with the well-known Fourier spectral representation. Section~\ref{sec:Isotropic} deals with isotropic propagation and shows that Clarke's model can be derived from the proposed model. A linear-system theoretic interpretation of general non-isotropic small-scale fading models is also provided.
Section~\ref{sec:Numerical} derives the Fourier plane-wave series expansion for a spatially-continuous compact space, which is used to numerically generate small-scale fading samples over Holographic MIMO arrays.
The numerical method is validated by means of Monte Carlo simulations.
Section~\ref{sec:conclusions} concludes with a general discussion and outlook on the developed analytical framework.

We will use upper (lower) case letters for frequency (time) entities. We use calligraphic letters for indicating sets and boldfaced lower case letters for vectors. 
$\mathbb{E}\{\cdot\}$ denotes the expectation operator. The notation ${n \sim \mathcal{N}_\mathbb{C}(0, \sigma^2)}$ stands for a circularly-symmetric Gaussian variable with variance $\sigma^2$. 
We use $\mathbb{R}^n$ to denote the $n$-dimensional real-valued vector space. 
$\nabla^2 = \frac{\partial^2}{\partial {x}^2} + \frac{\partial^2}{\partial {y}^2} + \frac{\partial^2}{\partial {z}^2}$ is the scalar Laplace operator. $\mathbbm{1}_{\mathcal{X}}(x)$ is the indicator function of a subset $x\in\mathcal{X}$. $\lceil x\rceil$ gives the smallest integer equal to or greater than $x$. $\delta(x)$ is the Dirac delta function while $\delta_n$ is the Kronecker delta.


 \vspace{-0.3cm}
\section{Physics-based Continuous Small-Scale Fading Modeling}  \label{sec:Model}
In the far-field of a homogeneous, isotropic, source-free, and {scattered} infinite medium, each of the three Cartesian components $(E_x,E_y,E_z)\in\Complex^3$ of the electrical field is a function of four scalar variables: the frequency $\omega$ (or, equivalently, the time $t$) and three Cartesian coordinates $(x,y,z)\in\Real^3$ denoting the spatial position.
Electromagnetic waves without polarization behave similarly to acoustic waves \cite{MarzettaIT}, and the small-scale fading 
can be modeled as a \emph{space-frequency scalar random field} \cite{PaulrajBook}
\begin{equation} \label{chan}
\left\{h_\omega(x,y,z):(x,y,z)\in\Real^3, \omega\in(-\infty,\infty)\right\}.
\end{equation}
A primary interest in wireless communications lies in scenarios in which $h_\omega$ can be modeled as a zero-mean, spatially-stationary, and Gaussian random field \cite{Marzetta2018}.
The Gaussian assumption is valid whenever 
the distance between scattering events 
is significant compared to the wavelength $\lambda$, but small compared to the distance between source and receiver \cite{Aris2000}. 
In this paper, we restrict our analysis to monochromatic waves, i.e., propagating at the same single frequency $\omega$. Therefore, we omit the subscript $_{\omega}$ in \eqref{chan} and call $h =h_\omega$. 
Following \cite{Baggeroer}, we pursue a continuous approach since discrete formulation may leave out significant physical underlying properties of wave propagation.

 \vspace{-0.3cm}
\subsection{Plane-Wave Solution}


\begin{figure*}[t!]
   \begin{minipage}[b]{.4\linewidth}
    \centering
     \includegraphics[width=\columnwidth]{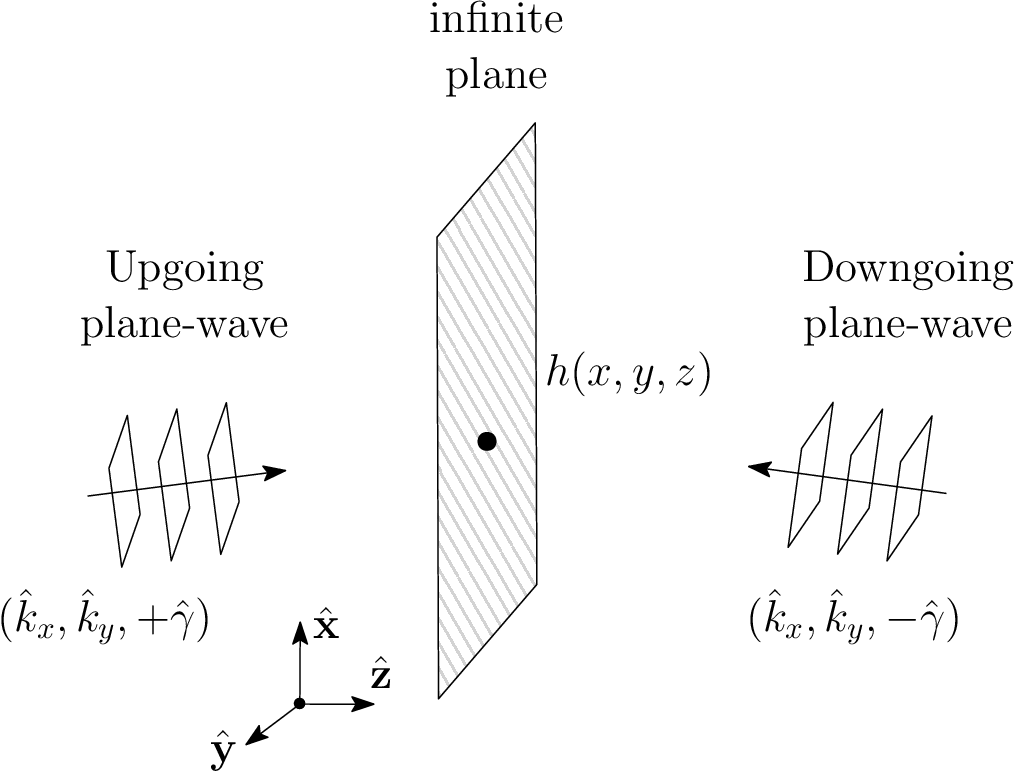} \vspace{-0.5cm}
     \subcaption{\vspace{-0cm}Deterministic LoS scenario.\vspace{-0.3cm}}\label{fig:Field_RX_LoS}
   \end{minipage}
   \hfill 
    \begin{minipage}[b]{.58\linewidth}
     \centering
     \includegraphics[width=\columnwidth]{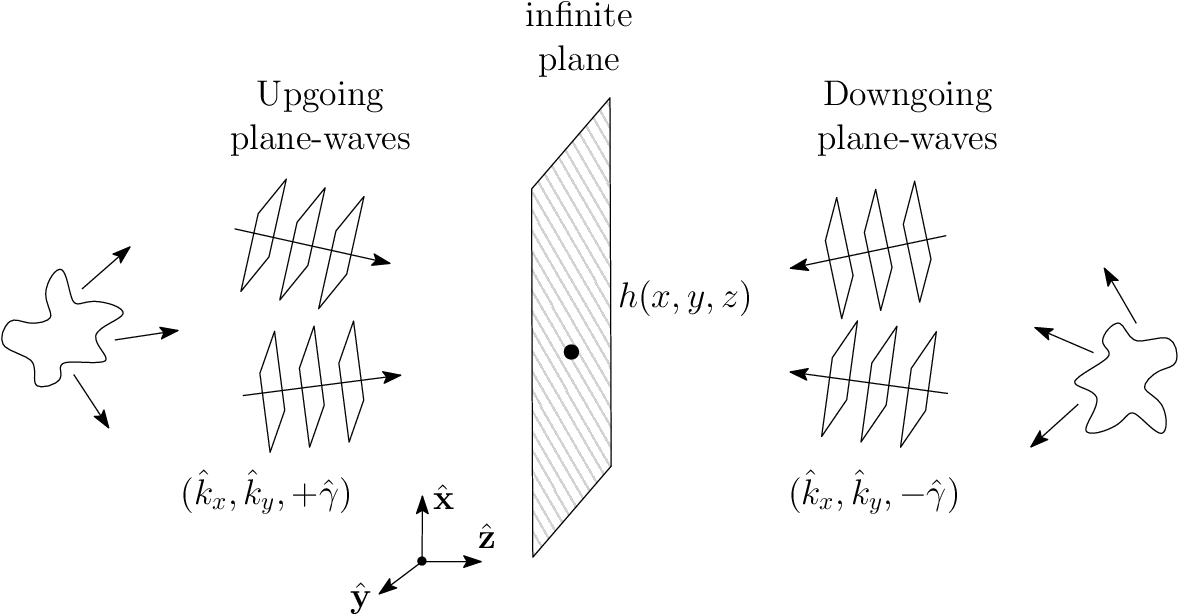} \vspace{-0.5cm}
     \subcaption{\vspace{-0cm}Random scattered NLoS scenario.\vspace{-0.3cm}}\label{fig:Field_RX}
   \end{minipage}
  \caption{Propagation of scalar plane-waves in a 3D environment.\vspace{-0.8cm}}
   \label{fig:Field}
\end{figure*}

The electromagnetic nature of the small-scale fading requires each realization of $h(x,y,z)$ to satisfy (with probability $1$) the scalar Helmholtz equation in the frequency domain. In a source-free environment, this means that \cite[Eq.~(1.2.17)]{ChewBook}:
\begin{equation} \label{eq:Helmholtz_frequency_mono} 
\left(\nabla^2 + \kappa^2  \right) h(x,y,z)  =  0
\end{equation}
where $\kappa = \frac{\omega}{c} = \frac{2 \pi}{\lambda}$ is the wavenumber (i.e., the angular displacement in radians per unit of length) 
with $c$ being the speed of light.
The Helmholtz condition in \eqref{eq:Helmholtz_frequency_mono} is a second-order linear partial differential equation with constant coefficients.
In analogy with ordinary differential equations, we are led to assume a solution with probability $1$ of the form \cite{HildebrandBook}
\begin{equation} \label{eq:solution}
h(x,y,z) = H e^{\imagunit (k_x x + k_y y + k_z z)} 
\end{equation}
where $(k_x,k_y,k_z)\in\Real^3$ and $H\in\Complex$ are unknown constant parameters.
A solution of this form is also motivated by the fact that \eqref{eq:Helmholtz_frequency_mono} is an eigenvalue equation in the Laplacian operator whose eigenfunction solutions can be obtained by inspection after recalling that the Laplacian operator $\nabla^2$ is linear and space-invariant \cite{HildebrandBook}.
The substitution of \eqref{eq:solution} into \eqref{eq:Helmholtz_frequency_mono} transforms it into an algebraic equation for all $(x,y,z)\in\Real^3$,
which yields the non-trivial solution $k_x^2+k_y^2+k_z^2 = \kappa^2$.
Thus, $k_z$ is determined, apart from a sign, by the other two components $k_x$ and $k_y$:
\begin{equation} \label{eq:kappa_z}
k_{z} = \pm \gamma(k_x,k_y) = \pm \sqrt{\kappa^2 - k_x^2 - k_y^2}.
\end{equation} 
Since we restrict the analysis to $k_z \in \Real$, it thus follows that $(k_x,k_y)$ must have compact support
\begin{equation} \label{eq:disk_wavenumber}
\mathcal{D}(\kappa) = \left\{(k_x,k_y)\in\Real^2 : k_x^2 + k_y^2 \le \kappa^2\right\}
\end{equation}
given by a disk of radius $\kappa$ centered on the origin. 
From \eqref{eq:solution} and \eqref{eq:kappa_z}, we thus have that 
\begin{equation} \label{eq:plane_wave}
\{H^+ e^{\imagunit (k_x x + k_y y + \gamma(k_x,k_y) z)}\} \quad \{H^- e^{\imagunit (k_x x + k_y y - \gamma(k_x,k_y) z)}\}
\end{equation}
are both distinct eigenfunctions of \eqref{eq:Helmholtz_frequency_mono}. 
As illustrated in Fig.~\ref{fig:Field_RX_LoS}, they describe two incident \emph{plane-waves}\footnote{Their wavefronts form infinite planes all oriented towards a propagating direction $(\hat{k}_x,\hat{k}_y,\pm\hat{\gamma}(k_x,k_y)) = (\frac{k_x}{\kappa},\frac{k_y}{\kappa},\pm\frac{\gamma(k_x,k_y) }{\kappa})$ \cite{ChewBook}.} impinging on the spatial point $(x,y,z)$ and spatially-propagating respectively through the left (upgoing) or right (downgoing) half-spaces created by an infinite plane passing through $(x,y,z)$ and perpendicular to the arbitrarily chosen $z-$axis.
We also notice that, by imposing the condition $(k_x,k_y,k_z)\in \Real^3$, we exclude the so-called \emph{evanescent waves} from the analysis (since they decay exponentially fast in space and contribute to the near-field propagation) and consider only \emph{propagating waves} as in Fig.~\ref{fig:Support_PSD_Disk}.

Notice that \eqref{eq:plane_wave} constitutes, for every $(k_x,k_y)\in \mathcal{D}(\kappa)$, the eigenspace solution to the Helmholtz equation and thus a solution to \eqref{eq:Helmholtz_frequency_mono} (with probability $1$) can be obtained as a linear combination of those eigenfunctions
\begin{equation} \label{eq:particular_solution}
h(x,y,z) = \iint_{\mathcal{D}(\kappa)} \Big(H^+(k_x,k_y) e^{\imagunit (k_x x + k_y y + \gamma(k_x,k_y) z)} +  H^-(k_x,k_y) e^{\imagunit (k_x x + k_y y - \gamma(k_x,k_y) z)} \Big) dk_xdk_y.
\end{equation}
In electromagnetic literature, this is known as a {general homogeneous plane-wave solution} to the Helmholtz equation (e.g., \cite[Ch.~6.7]{StrattonBook}),
where the wave amplitudes ${H^\pm(k_x,k_y) : \mathcal{D}(\kappa)\to\Complex}$ are now complex-valued functions taking arbitrary values within $\mathcal{D}(\kappa)$ so that \eqref{eq:particular_solution} is suitably convergent \cite{HildebrandBook}.
In other words, there are a possibly \emph{uncountably infinite} scattered waves that impinge on $(x,y,z)$ from directions ${(\hat{k}_x,\hat{k}_y,\pm\hat{\gamma})}$ and have complex-valued random amplitudes $H^\pm$. This effect is due to the interaction with the scattering environment, as showed in Fig.~\ref{fig:Field_RX}.
The scattered waves satisfy the Helmholtz equation independently, which in turn implies that \eqref{eq:particular_solution} satisfies it, due to the linearity of its operator.

The electromagnetic characterization of plane-wave amplitudes $H^\pm$ is done by using either ray-tracing modeling of the propagation environment \cite{Dardari2019a} or numerical methods for solving partial differential equations \cite{Salazar}. While these approaches account for specific geometries and are very accurate, they are scenario dependent and hard to work with mathematically.
We rather pursue a \emph{statistical approach} for modeling $H^\pm$ that leads to an analytically-tractable model.



\subsection{The Statistical Implication of the Helmholtz Equation}

For the zero-mean, spatially-stationary and Gaussian random field $h$, the spatial autocorrelation function
\begin{equation} \label{eq:ACF}
c_h (x,y,z) =  \Ex\{{h^*(x^\prime,y^\prime,z^\prime)} {h(x+x^\prime,y+y^\prime,z+z^\prime)}\}
\end{equation}
computed for every pair of points $(x^\prime,y^\prime,z^\prime)$ and $(x+x^\prime,y+y^\prime,z+z^\prime)$ provides a complete statistical description in the spatial domain.
Alternatively, $h$ can  be statistically described in the wavenumber domain by its power spectral density:
\begin{align} \label{eq:PSD}
S_h(k_x,k_y,k_z) = \iiint_{-\infty}^{\infty} \  c_h(x,y,z) \, e^{-\imagunit (k_x x + k_y y + k_z z)} \, dx dy dz 
\end{align}
from which $c_h$ follows as
\begin{align}   \label{eq:ACF_W}
c_h(x,y,z) =  \frac{1}{(2\pi)^3} \iiint_{-\infty}^{\infty}  S_h(k_x,k_y,k_z)  \, e^{\imagunit (k_x x + k_y y + k_z z)} \, dk_x dk_y dk_z.
\end{align}
Notice that the spatial and wavenumber (also known as \emph{spatial-frequency}) domains represent respectively the time and frequency counterparts of the classical Fourier analysis of time-domain signal \cite{Baggeroer}.
By applying the Laplacian operator to both sides of \eqref{eq:ACF} and interchanging the order of integration and derivation, we obtain
 \begin{equation} \label{eq:Helmholtz_pointsource_c_h}
\nabla^2  c_h(x,y,z) + \kappa^2 c_h(x,y,z) = 0
\end{equation}
which implies that the spatial autocorrelation itself satisfies the Helmholtz equation \cite{Marzetta2018}. 
By plugging \eqref{eq:ACF_W} into \eqref{eq:Helmholtz_pointsource_c_h} and interchanging the order of integration and differentiation yields 
\begin{equation} \label{eq:PSD_constraint}
\left(k_x^2 + k_y^2 + k_z^2 - \kappa^2\right) \, S_h(k_x,k_y,k_z) = 0.
\end{equation}
This implies that $S_h$ must vanish everywhere except on the wavenumber support 
\begin{equation} \label{eq:d_sphere}
\mathcal{S}(\kappa) = \{(k_x,k_y,k_z) \in \Real^3 : k_x^2 + k_y^2 + k_z^2 = \kappa^2\}
\end{equation}
of an impulsive sphere of radius $\kappa$  centered on the origin (see Fig.~\ref{fig:Support_PSD_dsphere}).
This is a consequence of the so-called \emph{spectral concentration} effect \cite{FranceschettiBook}, which follows from the fact that we consider an infinite propagation medium. 
Now, being $S_h$ defined over a support with zero measure, it cannot be interpreted as an ordinary function, but rather as a singular Delta distribution \cite{Johnson}.
Also, since any two distributions that are identical except for a set of zero measure returns the same Lebesgue integral \cite{GallagerBook},
the following result is established.
\begin{lemma}[\!\!\cite{Marzetta2018}] \label{lemma:PSD}
The power spectral density of any $h(x,y,z)$ obeying \eqref{eq:Helmholtz_frequency_mono} is in the form\footnote{The interpretation of a Dirac delta function having an argument which is a non-linear function is discussed in \cite{Marzetta99}.}
\begin{equation} \label{eq:PSD_function}  
S_h(k_x,k_y,k_z) = A_h^2\left({k_x},{k_y},{k_z}\right) \, \delta(k_x^2 + k_y^2 + k_z^2 - \kappa^2)
\end{equation}
where $A_h$ is a real-valued non-negative deterministic field called {spectral factor}, which model the spatial selectivity of the scattering.
\end{lemma}
Any \emph{physically-meaningful} small-scale fading must have a power spectral density of the form above. This is the statistical implication of the Helmholtz equation on the second-order statistics of $h$.
We notice that the i.i.d. Rayleigh fading model (i.e., obtained as a collection of independent zero-mean, circularly-symmetric Gaussian random variables for $(x,y,z)\in\Real^3$) is strictly speaking incompatible with the above result since it leads directly to an impulsive autocorrelation function $c_h(x,y,z) = S_0 \, \delta(x)\delta(y)\delta(z) $ and, in turn, to a constant power spectral density $S_h(k_x,k_y,k_z) = S_0$ for $(k_x,k_y,k_z)\in \Real^3$ according to \eqref{eq:PSD}. In Section IV, we will show that the model is, however, perfectly consistent with physics under isotropic propagation when the samples of the small-scale fading are taken along a straight line at a spacing of integer multiples of $\lambda/2$.


 \vspace{-0.3cm}
\section{Fourier Description of Physics-Based Channels} \label{sec:Fourier_spectral}

The power spectral density $S_h(k_x,k_y,k_z)$ of any physically meaningful small-fading $h(x,y,z)$ is defined in \eqref{eq:PSD_function} for the entire wavenumber spectrum $(k_x,k_y,k_z)\in\Real^3$. 
The power spectral density of the propagating waves \emph{only} can be obtained by taking the Fourier inversion of $S_h(k_x,k_y,k_z)$ with respect to $k_z \in \Real$ while the other two wavenumber components are held constant.

 \vspace{-0.3cm}
\subsection{Fourier inversion of $S_{h}(k_x,k_y,k_z)$ in $k_z$} 
From \eqref{eq:PSD}, we obtain
\begin{align}  \label{eq:Inversion_PSD_R3}
\frac{1}{2\pi} \int_{-\infty}^{\infty} S_{h}(k_x,k_y,k_z) e^{\imagunit k_z z} \,  dk_z  =  \iint_{-\infty}^{\infty}   c_h(x,y,z) \, e^{-\imagunit (k_x x + k_y y)} \, dx dy.
\end{align}
The composition of the Dirac delta with a differentiable function can be rewritten as \cite[Eq.~(181.a)]{ArfkenBook}
\begin{equation} \label{eq:Dirac_delta}
\delta\left(k_z^2 - (\kappa^2 - k_x^2 - k_y^2)\right) = \frac{\delta(k_z - \gamma) + \delta(k_z + \gamma)}{2 \, \gamma}
\end{equation}
where $\gamma$ is defined in \eqref{eq:kappa_z}.
We denote $A_{h,\pm}^2(k_x,k_y) =A_h^2(k_x,k_y,\pm \gamma)$ the values assumed by the spectral factor at $k_z = \pm \gamma$ along the axis $k_z$  (see Fig.~\ref{fig:Support_PSD_dsphere}). 
By using \eqref{eq:PSD_function} and \eqref{eq:Dirac_delta} into \eqref{eq:Inversion_PSD_R3} yields
\begin{align}  \label{eq:Inversion_PSD_R3_2}
\frac{1}{2\pi} \int_{-\infty}^{\infty} S_{h}(k_x,k_y,k_z) e^{\imagunit k_z z} \,  dk_z  = S_{h}^{+}(k_x,k_y) e^{\imagunit  \gamma z} +  S_{h}^{-}(k_x,k_y) e^{-\imagunit \gamma z}
\end{align}
which is the sum of the two 2D power spectral densities
\begin{equation} \label{eq:Spectral_Factor_projection_3d}
S_{h}^{\pm}(k_x,k_y) = \frac{A_{h,\pm}^2(k_x,k_y)/4\pi}{\sqrt{\kappa^2 - k_x^2 - k_y^2}}, \quad (k_x,k_y)\in\mathcal{D}(\kappa)
\end{equation}
where we have replaced $\gamma$ with its expression in \eqref{eq:kappa_z} and $(k_x,k_y)\in\mathcal{D}(\kappa)$ since we only considered propagating waves. The Fourier inversion of \eqref{eq:Inversion_PSD_R3_2} with respect to $(k_x,k_y)\in\mathcal{D}(\kappa)$ yields the autocorrelation function $c_h$ as given in \eqref{eq:Inversion_PSD_R3}. 

This intermediate result is instrumental to obtain later a statistical characterization of the propagating wave amplitudes $H^\pm$ in \eqref{eq:particular_solution} that is based on $S_{h}^{\pm}$. For this reason, \eqref{eq:Spectral_Factor_projection_3d} are referred to as \emph{plane-wave spectrums}. 
One should not be worried about the singularity at the boundary of their spectral support $\mathcal{D}(\kappa)$. If the spectral factor $A_{h}$ is bounded on $\mathcal{D}(\kappa)$, this singularity can be removed by applying a change of integration variables (as shown in Appendix~\ref{app:numerical_generation_part2} for a constant $A_{h}$), which makes $S_{h}^{\pm}$ singularly-integrable spectrums.

As illustrated in Fig.~\ref{fig:Support_PSD_dsphere}, $S_{h}^{\pm}$ are obtained from $S_{h}$ by independently parametrizing the upper and lower hemispheres of $\mathcal{S}(\kappa)$ on the 2D wavenumber disk support $\mathcal{D}(\kappa)$, respectively. 
Their magnitudes are driven by the Jacobian determinant of the spherical parametrization ${\varphi(k_x,k_y): \mathcal{D}(\kappa) \to \mathcal{S}(\kappa)}$ induced by the $k_z$-Fourier inversion
\begin{equation}
J_\varphi(k_x,k_y) = \sqrt{\left(\frac{\partial \gamma}{\partial k_x}\right)^2 + \left(\frac{\partial \gamma}{\partial k_y}\right)^2 +1}  
= \frac{\kappa}{\sqrt{\kappa^2 - k_x^2 - k_y^2}} 
\end{equation}
which leads to large values near the boundary of $\mathcal{D}(\kappa)$.


\begin{figure*}[t!]
   \begin{minipage}[b]{.5\columnwidth}
    \centering
     \includegraphics[width=.85\columnwidth]{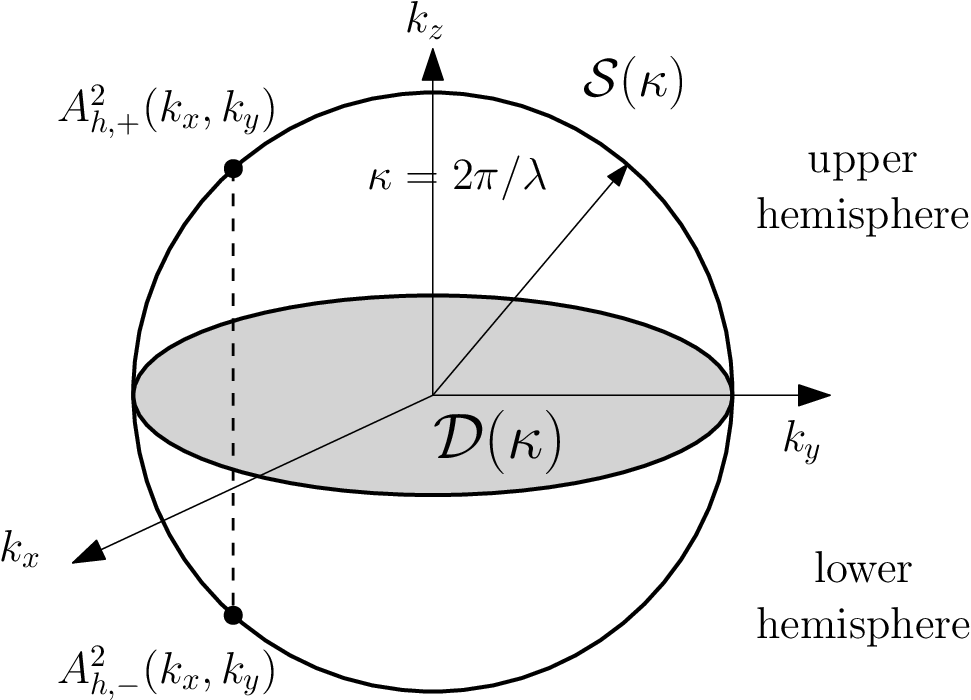} \vspace{-0cm}
     \subcaption{\vspace{-0cm}{Spectral sphere $\mathcal{S}(\kappa)$ and its 2D parametrization $\mathcal{D}(\kappa)$.}\vspace{-0.2cm}}\label{fig:Support_PSD_dsphere}
   \end{minipage}
   \hfill 
    \begin{minipage}[b]{.5\columnwidth}
     \centering
     \includegraphics[width=.63\columnwidth]{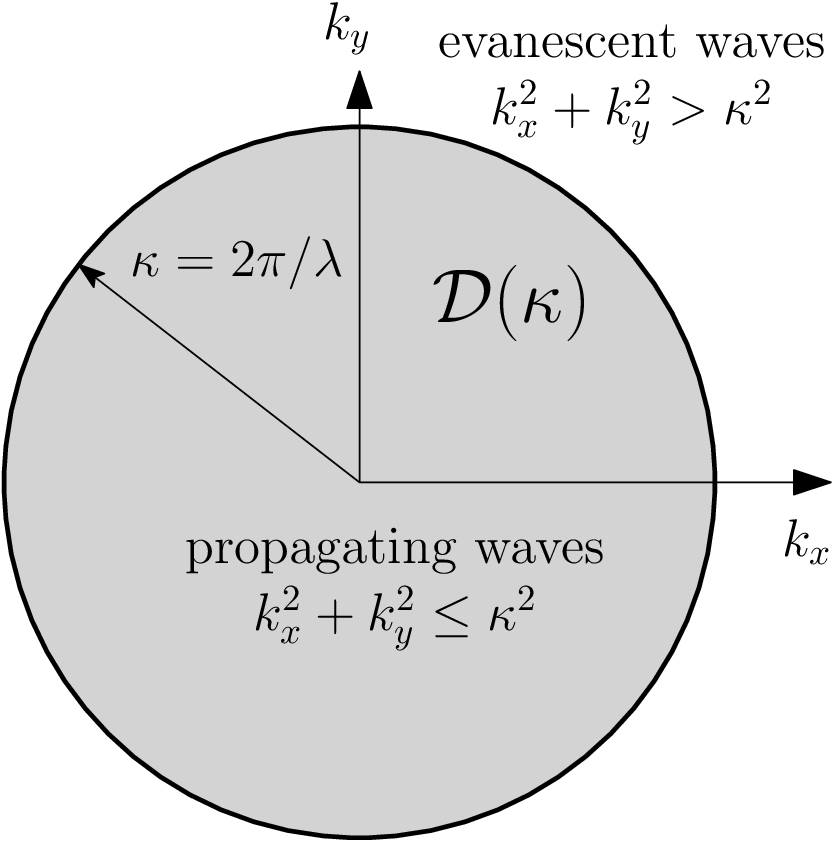} \vspace{-0cm}
     \subcaption{\vspace{-0cm}{Propagating and evanescent waves.}\vspace{-0.2cm}}\label{fig:Support_PSD_Disk}
   \end{minipage}
  \caption{Wavenumber support of the power spectral density of the channel $h(x,y,z)$. \vspace{-0.8cm}}
   \label{fig:Support_PSD}
\end{figure*}
\subsection{The 2D Fourier Plane-Wave Spectral Representation}

From \eqref{eq:particular_solution}, to generate a zero-mean, spatially-stationary and Gaussian random field model for $h$ we can  choose the amplitudes $H^\pm$ independently as a 2D collection of \emph{statistically-independent}, circularly-symmetric Gaussian random variables
\begin{equation} \label{eq:amplitudes}
H^\pm(k_x,k_y) =  \sqrt{S_{h}^{\pm}(k_x,k_y)} W^\pm(k_x,k_y)
\end{equation}
where $S_{h}^{\pm}(k_x,k_y)$ are the {plane-wave spectrums} in \eqref{eq:Spectral_Factor_projection_3d} and $W^\pm$ are two 2D independent, zero-mean, complex-valued, white-noise Gaussian random fields with unit variance. Spatial stationarity can thus be directly verified.
Physically, $S_{h}^{\pm}(k_x,k_y)$ are associated with upgoing and downgoing propagating waves in \eqref{eq:particular_solution} and for every $(k_x,k_y)\in\mathcal{D}(\kappa)$ they specify the average power carried by these waves at each direction $(k_x,k_y,\pm\gamma)$. The white-noise fields capture the randomness of the scattering propagation environment.

The substitution of the propagating wave amplitudes \eqref{eq:amplitudes} into the homogeneous solution \eqref{eq:particular_solution} yields the following spectral representation for the 3D small-scale fading \cite{Marzetta2018}
\begin{align} \label{eq:small_scale_R3}
h(x,y,z) 
=  h_+(x,y,z)+h_-(x,y,z)
\end{align} 
with
\begin{equation}  \label{eq:small_scale_R3_halfspace}
h_\pm(x,y,z) = \frac{1}{4\pi\sqrt{\pi}}  \iint_{\mathcal{D}(\kappa)}  \frac{A_{h,\pm}(k_x,k_y)}{(\kappa^2 - k_x^2 - k_y^2)^{1/4}} \, W^\pm(k_x,k_y)  \; e^{\imagunit \left(k_x x + k_y y \pm \sqrt{\kappa^2 - k_x^2 - k_y^2} \, z\right)} \, d k_x d k_y
\end{equation}
where we have used \eqref{eq:kappa_z}. 
Thus, $h(x,y,z) $ is \emph{exactly} expressed as a linear superposition of possibly an uncountably infinite number of upgoing and downgoing propagating waves having statistically-independent Gaussian-distributed random amplitudes. This is the only physically-tenable model for describing any arbitrary spatially-stationary Gaussian random small-scale fading.
Notice that the convergence of \eqref{eq:small_scale_R3_halfspace} is always guaranteed in the mean-square sense since $S_h^\pm(k_x,k_y)$ are singularly-integrable over $\mathcal{D}(\kappa)$ for every upper-bounded spectral factor $A_{h,\pm}(k_x,k_y)$; see Appendix~IV.C.

Also, \eqref{eq:small_scale_R3_halfspace} is reminiscent of a 2D inverse Fourier transform in the spatial variables $(x,y)$, with fixed $z$, where its Fourier harmonics (normalized) physically correspond to propagating waves. In other words, the Fourier transform acts as a continuous {plane-wave decomposition} of the channel, and, for this reason, we generally refer to \eqref{eq:small_scale_R3} and \eqref{eq:small_scale_R3_halfspace} as a \emph{Fourier plane-wave spectral representation} of $h$.
The connection with Fourier theory will be instrumental in deriving in Section~\ref{sec:Numerical} a Fourier plane-wave series expansion over compact spaces, from which a numerical procedure is derived to generate spatial samples of the small-scale fading over apertures of compact size. 

\subsection{Connection with the Fourier Spectral Representation} \label{sec:Spectral}

Statistical integral representations of stationary random processes are available in signal processing. The most important is the Fourier spectral representation \cite[Sec.~3.6]{VanTreesBook}, which can be regarded as the asymptotic version of the famous Karhunen-Loeve series expansion \cite[Sec.~3.2]{VanTreesBook}.   
In Appendix~I, we revise this representation and obtain a handy expression in  \eqref{eq:Spectral_representation_PSD} to work with.
The generalization of this theory to 3D spatially-stationary random fields is treated in \cite[Eq.~(2.12)]{Baggeroer}, which leads to a similar form as \eqref{eq:Spectral_representation_PSD} given by: 
\begin{equation} \label{eq:Spectral_representation_PSD_Rd}
h(x,y,z)  =  \frac{1}{(2\pi)^{3/2}} \iiint_{-\infty}^{\infty}  \sqrt{S_h(k_x,k_y,k_z)} W(k_x,k_y,k_z) e^{\imagunit (k_x x + k_y y + k_z z)} \, dk_x dk_y dk_z 
\end{equation}
where $W$ is a 3D zero-mean, stationary, white-noise Gaussian random field with unit spectrum, $S_h$ is the power spectral density of the random field $h$. 
The direct evaluation of $h$ through \eqref{eq:Spectral_representation_PSD_Rd} requires the computation of the square-root of the impulsive function $S_h$ in \eqref{eq:PSD_function}, which makes the 3D Fourier spectral representation in \eqref{eq:Spectral_representation_PSD_Rd} inadequate to statistically describe physics-based channels.\footnote{Square-root of singular functions are not defined even in the sense of distributions \cite{Johnson}.}

Notice, however, that evaluating \eqref{eq:small_scale_R3} on the infinite plane $z= 0$ yields $h(x,y,0) = h_{+}(x,y,0) + h_{-}(x,y,0)$
with 
\begin{equation} \label{eq:small_scale_R3_z0}
h_{\pm}(x,y,0) = \frac{1}{2\pi} \iint_{-\infty}^{\infty}   \sqrt{S_{h}^{\pm}(k_x,k_y)} W^\pm(k_x,k_y) e^{\imagunit (k_x x + k_y y)}  \, d k_x d k_y
\end{equation}   
where $S_{h}^{\pm}$ and $W^\pm$ are defined in \eqref{eq:amplitudes}. 
The above expression coincides with the 2D Fourier spectral representation.
By using the theory of linear-systems (reviewed in Appendix~I and summarized in Fig.~\ref{fig:LTI_2}), it thus follows that $h_{\pm}(x,y,0)$ can be generated by passing $W^\pm(k_x,k_y)$ through 2D linear space-invariant filters with wavenumber responses $\sqrt{S_{h}^{\pm}(k_x,k_y)}$. 
These physical filters are \emph{bandlimited} with wavenumber bandwidth $|\mathcal{D}(\kappa)| = \pi \kappa^2$ and of second-order having two poles $k_z = \pm \gamma(k_x,k_y)$ on the $k_z-$axis that correspond to the plane-wave spectrum singularities.
This filtering operation is due to the Helmholtz equation \eqref{eq:Helmholtz_frequency_mono}.
Indeed, the Helmholtz operator $\nabla^2+\kappa^2$ is linear and space-invariant; that is, the linear combination of any two solutions is a solution, and a space-shifted version of $h$ is also a solution, and wave propagation can be treated by using the theory of linear systems.
Finally, the physics-based channels $h_\pm(x,y,z)$, evaluated at any infinite plane $z\ne0$, can be obtained as a space-shifted version of $h_{\pm}(x,y,0)$ by passing it through two phase-shift filters with wavenumber responses $e^{\pm\imagunit \gamma z}$; $h(x,y,z)$ is eventually obtained as the sum of the two filter outputs, as shown in Fig.~\ref{fig:GenerationSmallScale}.
These filters are known as \emph{migration filters} in the geophysical literature \cite{Marzetta2018} and describe lossless wave propagation through the left and right half-spaces.

\begin{figure}[t!]
     \centering
     \includegraphics[width=.8\columnwidth]{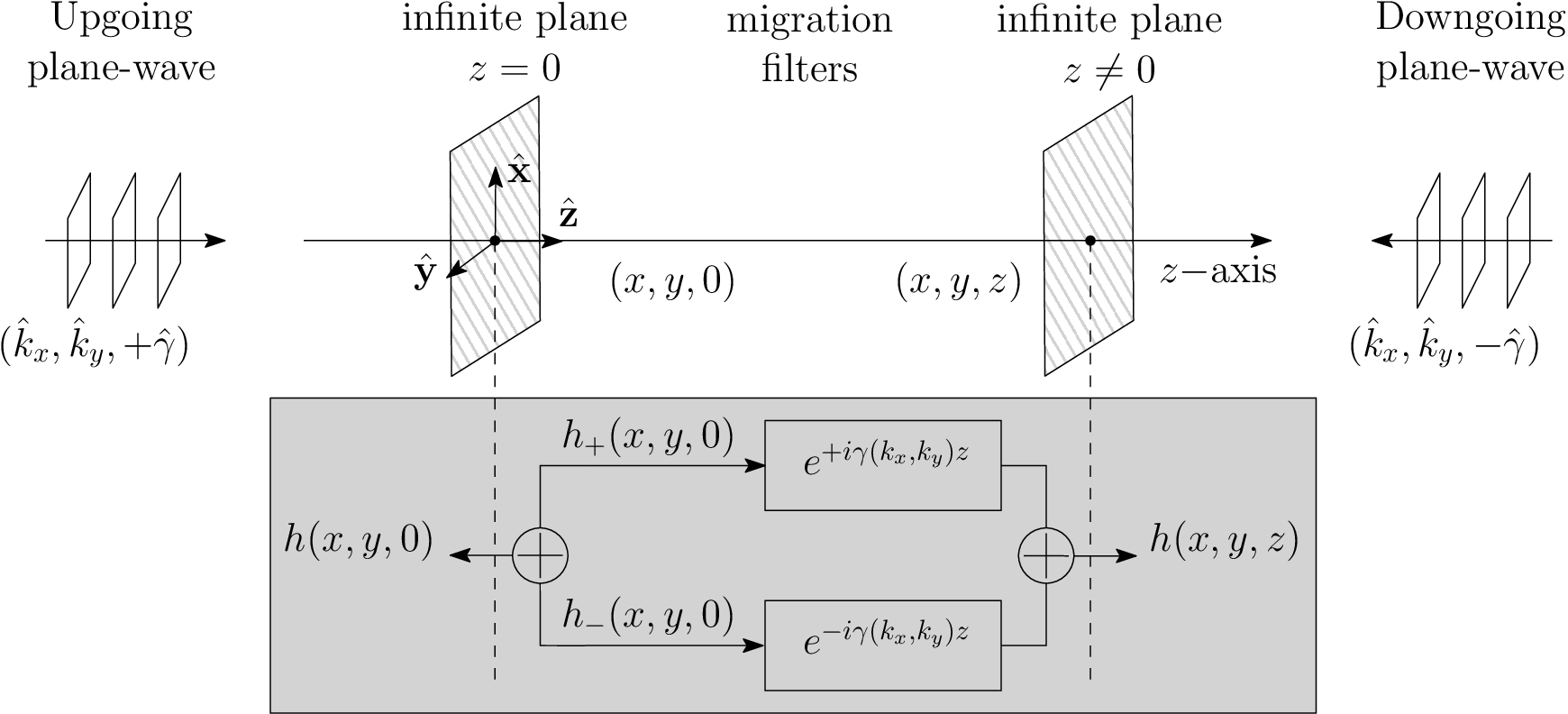} \vspace{-0cm}
  \caption{Extrapolation of the small-scale fading over infinite $z-$planes through migration filters.\vspace{-0.8cm}}
   \label{fig:GenerationSmallScale}
\end{figure}

 \vspace{-0.2cm}
\section{Isotropic Propagation} \label{sec:Isotropic}
We now discuss the connection between the proposed model and the Clarke's model. Particularly, we show that the latter is the closest physically-tenable model for i.i.d. Rayleigh fading and that every non-isotropic channel can be generated by passing a Clarke's isotropic channel through a linear space-invariant filter.
\begin{definition}[\!\!\cite{PaulrajBook}]
An isotropic channel, say $\tilde{h}(x,y,z)$, is characterized by a radially-symmetric spectral factor, which is invariant under rotations. 
\end{definition}
This implies that we can choose the wavenumber coordinates in a convenient way such that they are aligned to one of their axes, say the $k_x-$axis. We may thus write
\begin{equation} \label{eq:spectral_factor3D}
A_{\tilde h}\left({k_x},{k_y},{k_z}\right) = A_{\tilde h}\left(\sqrt{k_x^2 + k_y^2 + k_z^2}, 0 ,0\right) = A_{\tilde h}(\kappa)
\end{equation}
where $A_{\tilde h}$ is the spectral factor associated with the isotropic channel. By assuming that the overall power of $\tilde{h}$ is normalized to $1$, in Appendix~II it is shown that 
\begin{equation}\label{eq:spectral_factor3D_closed_form}
A_{\tilde h}(\kappa) = \frac{2\pi}{\sqrt{\kappa}}.
\end{equation}
By using \eqref{eq:spectral_factor3D} into \eqref{eq:PSD_function}, the power spectral density of $\tilde h$ becomes
\begin{equation} \label{eq:spectral_density3D}
S_{\tilde h}\left({k_x},{k_y},{k_z}\right) = \frac{4\pi^2}{{\kappa}}\delta(k_x^2 + k_y^2 + k_z^2 - \kappa).
\end{equation}
Substituting \eqref{eq:spectral_factor3D_closed_form} into \eqref{eq:small_scale_R3} and \eqref{eq:small_scale_R3_halfspace} yields $\tilde{h}(x,y,z) = \tilde{h}_+(x,y,z) + \tilde{h}_-(x,y,z)$ with
\begin{equation}  \label{eq:small_scale_R3_halfspace_ISO}
\tilde{h}_\pm(x,y,z) = \frac{1}{2\sqrt{\pi \kappa}}  \iint_{\mathcal{D}(\kappa)}  \frac{e^{\imagunit \left(k_x x + k_y y \pm \sqrt{\kappa^2 - k_x^2 - k_y^2} z\right)}}{(\kappa^2 - k_x^2 - k_y^2)^{1/4}} W^\pm(k_x,k_y)  d k_x d k_y
\end{equation}
where $W^\pm$ are two independent 2D Gaussian white noise fields with unit variance.
Although \eqref{eq:spectral_factor3D} implies that each propagating wave carries equal power, after the Fourier inversion along $\kappa_z$, the power carried by each wave in \eqref{eq:small_scale_R3_halfspace_ISO} is no longer constant. This is due to the spherical parametrization discussed above for the generic non-isotropic case.
The evaluation of $\tilde{h}$ at $z= 0$ yields  the 2D Fourier spectral representation $\tilde{h}(x,y,0) = \tilde{h}_+(x,y,0) + \tilde{h}_-(x,y,0)$ with
\begin{equation} \label{eq:small_scale_R3_ISO_z0}
\tilde{h}_\pm(x,y,0) = \frac{1}{2\pi} \iint_{-\infty}^{\infty} \sqrt{S_{\tilde{h}}(k_x,k_y)} W^\pm(k_x,k_y) e^{\imagunit \left(k_x x + k_y y\right)} \, d k_x d k_y
\end{equation} 
where the power spectral density is of the form
\begin{equation} \label{eq:Spectral_Factor_projection_3d_isotropic}
S_{\tilde{h}}(k_x,k_y) =   \frac{\pi/\kappa}{\sqrt{\kappa^2 - k_x^2 - k_y^2}}, \quad (k_x,k_y)\in\mathcal{D}(\kappa)
\end{equation}
which is a bandlimited singularly-integrable spectrum that guarantees convergence of the integral representation \eqref{eq:small_scale_R3_ISO_z0}; see Appendix~IV.C.
The autocorrelation function is available in closed-form. 
%
 \begin{lemma} \label{ACF_Jakes}
 If the channel is isotropic with $A_{\tilde h}(\kappa) = {2\pi}/{\sqrt{\kappa}}$, then the 3D inverse Fourier transform of the spectrum yields the explicit autocorrelation function\footnote{The functional dependence on the distance $r$ between every pair of points is a standard property of isotropic random fields \cite{PaulrajBook}.}
   \begin{equation}\label{eq:Jakes3D}
c_{\tilde{h}}(x,y,z) = \sinc\left(\frac{2 r}{\lambda}\right)=\frac{\sin(\kappa r)}{\kappa r}
\end{equation}
where $r=\sqrt{x^2 + y^2 + z^2}$ is the distance among any pair of spatial points.
 \end{lemma}
 \begin{proof}
The proof is given in Appendix~III. 
\end{proof}
Since $\frac{\sin(\kappa r)}{\kappa r}$ is zero when $r = k \lambda/2$ with $k\in \Integer$, samples of the isotropic small-scale fading taken along a straight line at a spacing of an integer multiple of $\lambda/2$ are independent. This result holds true also when $\lambda\to 0$ ($\kappa\to \infty$).
\begin{corollary} \label{sinc_approximation}
The samples of the isotropic small-scale fading become asymptotically independent as $\lambda\to 0$ ($\kappa\to \infty$) with
\begin{equation}
c_{\tilde{h}}(x,y,z) = \delta(r)=\delta(x)\delta(y)\delta(z).
\end{equation}
 \end{corollary}
  \begin{proof}
The Dirac delta function may be represented through the normalized $\sinc(\cdot)$ function, i.e., ${\lim_{a\to0} \frac{\sinc({r}/{a})}{a} \to \delta(r)}$ with ${a = \lambda/2}$, where the limits must be intended in distribution sense.
\end{proof}
From Lemma~\ref{ACF_Jakes} and Corollary~\ref{sinc_approximation}, we can conclude that that the isotropic small-scale fading ${\tilde h}(x,y,z)$ is the closest physically-tenable model to the i.i.d. Rayleigh fading, which is thus perfectly consistent with physics principles under the conditions above.

There is a 2D counterpart to the 3D theory presented in this paper, where the small-scale fading depends only on $(x,y)$. In this case, the wavenumber support of the power spectral density is an impulsive circle of radius $\kappa$ (compare this to the 3D spectral support in Fig.~\ref{fig:Support_PSD}). 
Under isotropic propagation, by choosing $A_{\tilde h}(\kappa) = 2\sqrt{\pi}$ for normalization purpose (see Appendix~II), we obtain $\tilde{h}(x,y) = \tilde{h}_+(x,y) + \tilde{h}_-(x,y)$ with
\begin{equation}  \label{eq:small_scale_R2_halfspace_ISO}
\tilde{h}_\pm(x,y) = \frac{1}{\sqrt{2 \pi}}  \int_{-\kappa}^\kappa  \frac{e^{\imagunit \left(k_x x \pm \sqrt{\kappa^2 - k_x^2} \, y\right)}}{(\kappa^2 - k_x^2)^{1/4}} W^\pm(k_x)  \, dk_x
\end{equation}
where $W^\pm$ are two independent 1D Gaussian white noise fields with unit variance.
At $y=0$, the 1D Fourier spectral representation reads as $\tilde{h}(x,0) = \tilde{h}_+(x,0) + \tilde{h}_-(x,0)$ with
\begin{equation} \label{eq:small_scale_R2_ISO}
\tilde{h}_\pm(x,0) = \frac{1}{\sqrt{2\pi}} \int_{-\infty}^{\infty}   \sqrt{S_{\tilde{h}}(k_x)} W^\pm(k_x) e^{\imagunit k_x x} \, d k_x 
\end{equation} 
and bandlimited singularly integrable spectrum
\begin{equation} \label{eq:Spectral_Factor_projection_2d}
S_{\tilde{h}}(k_x) = \frac{1}{\sqrt{\kappa^2 - k_x^2}}, \quad k_x\in[-\kappa,\kappa].
\end{equation} 
The autocorrelation function of $\tilde{h}$ can be computed in closed-form. 
\begin{lemma} \label{ACF_Jakes2D}
 If the channel is isotropic with $A_h^2(\kappa) = 4\pi$ (see Appendix~\ref{app:power_normalization}), then the autocorrelation function is
   \begin{equation} \label{eq:Jakes2D}
c_{\tilde{h}}(x,y) = J_0\left(\frac{2\pi r}{\lambda}\right)
\end{equation}
where $r=\sqrt{x^2 + y^2}$, and $J_0\left(x\right)$ is the Bessel function of first kind and order $0$. 
 \end{lemma}
  \begin{proof}
The proof is sketched at the end of Appendix~III. 
\end{proof}
Lemma~\ref{ACF_Jakes} and Lemma~\ref{ACF_Jakes2D} show that $c_{\tilde{h}}$ is the same autocorrelation function obtained by respectively using the 3D \cite{Aulin79} and 2D \cite{Clarke} Clarke's models.
The same results are obtained by using a ``diffusion approximation'' in \cite{Aris2000}. Hence, the tools developed so far allow to study small-scale fading in its most general form, while being in agreement with previous models.

\subsection{Linear System-Theoretic Interpretation of Scattering}

By using \eqref{eq:Spectral_Factor_projection_3d_isotropic}, we can rewrite \eqref{eq:Spectral_Factor_projection_3d}
 as follows
\begin{equation} \label{eq:PSD_isotropic_Jakes}
S_{h}^{\pm}(k_x,k_y) = \left(\frac{A_{h,\pm}^2(k_x,k_y)}{4\pi^2/\kappa} \right)  \, S_{\tilde{h}}(k_x,k_y)
\end{equation}
which implies that we can generate any channel $h(x,y,0)$ with an arbitrary spectrum by passing an isotropic channel $\tilde{h}(x,y,0)$ through a 2D linear space-invariant filter with wavenumber response given by the spectral factor (up to a normalization factor). Hence, the term between brackets in \eqref{eq:PSD_isotropic_Jakes} can be interpreted as the spatial frequency response of a shaping-filter that turns $\tilde{h}(x,y,0)$ into $h(x,y,0)$; see Fig.~\ref{fig:LTI}.
This is the system-theoretic importance of the isotropic model for generating any random scattering channel.
Then of ${h}(x,y,z)$ and $\tilde{h}(x,y,z)$ can always be obtained by passing their versions evaluated at the infinite plane $z=0$ through the corresponding migration filters $e^{\pm\imagunit \gamma z}= e^{ \pm\imagunit  \sqrt{\kappa^2 - k_x^2 - k_y^2} z}$ as shown in Fig.~\ref{fig:GenerationSmallScale}.


\begin{figure}[t!]
    \centering
     \includegraphics[width=.45\columnwidth]{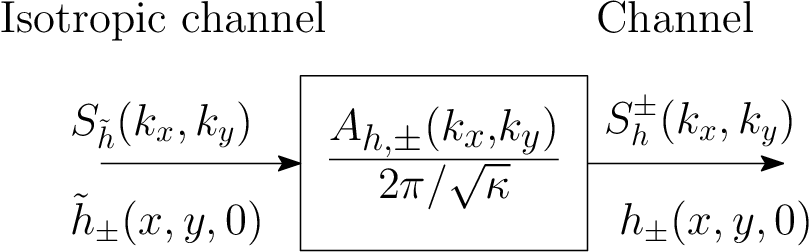}
     \caption{\vspace{-0cm}Linear system-theoretic interpretation of {wave propagation through scattering.} \vspace{-0.8cm}}
     \label{fig:LTI}
\end{figure}

 \vspace{-0.3cm}
\section{Discrete Representation of the Continuous Model} \label{sec:Numerical}
{In the above sections, we have shown that every non-isotropic small-scale fading $h(x,y,z)$ that satisfies the Helmholtz equation, is a bandlimited second-order random field with wavenumber support $\mathcal{D}(\kappa)$.
The small-scale fading $h(x,y,0)$ observed at the infinite plane $z=0$ can be obtained from a bandlimited isotropic small-scale fading $\tilde{h}(x,y,0)$ through a linear space-time invariant filtering operation (see Fig.~\ref{fig:LTI}), which is in turn described by the 2D Fourier spectral representation in \eqref{eq:small_scale_R3_ISO_z0}.
The evaluation of $h(x,y,z)$ at any infinite $z-$plane is implemented by passing $h(x,y,0)$ through migration filters, as shown in Fig.~\ref{fig:GenerationSmallScale}.
All this is used next to develop a numerical procedure to efficiently generate channel samples over compact rectangular apertures (i.e., linear, planar, and volumetric apertures).
}


\subsection{Fourier Plane-Wave Series Expansion of Isotropic Channels}

Consider a compact spatially-continuous rectangular space $\mathcal{V} = \{(x,y)\in\Real^2: x\in[0,L_x], y\in [0,L_y]\}$ of side lengths $L_x,L_y<\infty$.
The same procedure described in Appendix~IV.A to obtain a Fourier spectral series expansion approximation of a bandlimited random process with $\omega\in[-\Omega,\Omega]$ defined over $t\in[0,T]$ can be extended to isotropic spatial random fields observed over a space region $\mathcal{V}$ of finite size. This is obtained by replacing the time-frequency mapping with its space-wavenumber counterpart: the time interval $t\in[0,T]$ must be replaced with the spatial region $(x,y)\in\mathcal{V}$, and the angular-frequency interval $\omega\in[-\Omega,\Omega]$ with the wavenumber region $(k_x,k_y)\in\mathcal{D}(\kappa)$.
In analogy with Appendix~IV.A, we partition the spectral support $\mathcal{D}(\kappa)$ in Fig.~\ref{fig:Support_PSD_Disk} uniformly with  spacing $\Delta_{k_x} = 2\pi/L_x$ and $\Delta_{k_y} = 2\pi/L_y$ along the $k_x$ and $k_y-$axes.
By rescaling the $k_x$ and $k_y-$axes as $\frac{L_x}{2\pi} k_x$ and $\frac{L_y}{2\pi} k_y$, these partitions are indexed by
\begin{equation} \label{ellipse}
\mathcal{E} = \{(\ell,m)\in\Integer^2 : \left({\ell \lambda}/{L_x}\right)^2 + \left({m \lambda}/{L_y}\right)^2 \le 1\}
\end{equation}
which is a 2D lattice ellipse of semi-axes $L_x/\lambda$ and $L_y/\lambda$, as shown in Fig.~\ref{fig:disk_lattice}. The white and red dots indicate the wavenumber harmonics associated with propagating and evanescent plane-waves, respectively.
In the same way \eqref{eq:Spectral_representation_PSD} can be approximated over $t\in[0,T]$ by \eqref{eq:DTFT} as $\Omega T \to \infty$, \eqref{eq:small_scale_R3_ISO_z0} can be approximated over $(x,y)\in\mathcal{V}$ 
as $\min(L_x,L_y)/\lambda\to\infty$ by  
\begin{equation} \label{eq:Fourier_series_planar}
\tilde{h}(x,y) \approx  \mathop{\sum\sum}_{(\ell,m)\in \mathcal{E}} \tilde{H}_{\ell m} e^{\imagunit 2\pi \left(\frac{\ell x}{L_x} + \frac{m y}{L_y}\right)}, \quad (x,y)\in \mathcal{V} 
\end{equation}
where $\tilde{H}_{\ell m} \sim \CN(0, 2 \sigma_{\ell m}^2)$ are statistically-independent Gaussian-distributed random variables with variances $\sigma_{\ell m}^2$ as computed in Appendix~\ref{app:numerical_generation_part2}. 
In analogy with Fourier theory, we refer to \eqref{eq:Fourier_series_planar} as the \emph{Fourier plane-wave series expansion} of $\tilde{h}(x,y)$ over $(x,y)\in\mathcal{V}$.
The channel energy collected over the finite spatial region $\mathcal{V}$ is contained in a \emph{countably-finite} number of angular directions each one of which corresponding to a different propagating wave.  
In analogy to the Fourier series of a time-domain waveform, \eqref{eq:Fourier_series_planar} provides a periodic spatial random field with 2D fundamental period $(x,y)\in\mathcal{V}$. 
The Fourier plane-wave series expansion of $\tilde{h}(x,y,z)$ at different $z$ can be implemented by using migration filters having wavenumber responses $e^{\pm\imagunit \gamma_{\ell m} z}$ with
\begin{equation}
\gamma_{\ell m}= \sqrt{\kappa^2 - \left(\frac{2\pi\ell}{L_x}\right)^2 - \left(\frac{2\pi m}{L_y}\right)^2} 
= \kappa \sqrt{1 -\left( \frac{\ell}{L_x/\lambda}\right)^2 -\left( \frac{m}{L_y/\lambda}\right)^2}
\end{equation}
as obtained from \eqref{eq:kappa_z} by evaluating $(\kappa_x, \kappa_y)$ at $\left(\frac{2\pi\ell}{L_x}, \frac{2\pi m}{ L_y}\right)$.
The representation in \eqref{eq:Fourier_series_planar} becomes\footnote{Since the model accuracy increases as $\min(L_x,L_y)/\lambda$ grows large, we have that $|z|<\min(L_x,L_y)$. Otherwise one may always find a better approximation by exchanging the $z-$axis with one of the other two.}
\begin{equation} \label{eq:Fourier_series_volumetric}
\tilde{h}(x,y,z) \approx  \mathop{\sum\sum}_{(\ell,m)\in \mathcal{E}} \tilde{H}_{\ell m}(z) e^{\imagunit 2\pi \left(\frac{\ell x}{L_x} + \frac{m y}{L_y}\right)}, \quad (x,y)\in \mathcal{V}, |z|<\min(L_x,L_y)
\end{equation}
where 
\begin{equation} \label{eq:Fourier_coeff_3D}
\tilde{H}_{\ell m}(z)  = \tilde{H}_{\ell m}^+  e^{\imagunit \gamma_{\ell m} \; z}  +  \tilde{H}_{\ell m}^- e^{-\imagunit \gamma_{\ell m} \; z}
\end{equation} 
and $\tilde{H}_{\ell m}^\pm \sim \CN(0, \sigma_{\ell m}^2)$. 
The restriction to a linear aperture of length $L_x<\infty$ squeezes the 2D spectral disk $\mathcal{D}(\kappa)$ in Fig.~\ref{fig:Support_PSD_Disk} to the $k_x-$axis so that the wavenumber support becomes a 1D segment $k_x\in[-\kappa,\kappa]$ as shown in \eqref{eq:Spectral_Factor_projection_2d}.
By partitioning the spectral segment uniformly with wavenumber spacing interval $2\pi/L_x$ so that each partition is indexed by $\ell= \{-L_x/\lambda,\ldots, L_x/\lambda-1\}$ (see Fig.~\ref{fig:disk_lattice}), the Fourier plane-wave series expansion of $\tilde{h}(x)=\tilde{h}(x,0,0)$ over $x\in[0,L_x]$ becomes a 1D inverse discrete Fourier transform:
\begin{equation} \label{eq:Fourier_series_linear}
\tilde{h}(x) \approx  \sum_{\ell = -L_x/\lambda}^{L_x/\lambda-1} \tilde{H}_{\ell} \, e^{\imagunit 2\pi \frac{\ell x}{L_x}}, \quad x\in[0,L_x]
\end{equation}
where $\tilde{H}_{\ell} \sim \CN(0, 2\sigma_{\ell}^2)$ are statistically-independent Gaussian-distributed random variables with variances $\sigma_{\ell}^2$ as computed in Appendix~IV.C. 
Next, we use \eqref{eq:Fourier_series_volumetric} and \eqref{eq:Fourier_series_linear} to develop a numerical procedure to efficiently generate channel samples over compact rectangular apertures (i.e., linear, planar, and volumetric apertures).

\begin{figure}[t!]
     \centering
     \includegraphics[width=0.5\columnwidth]{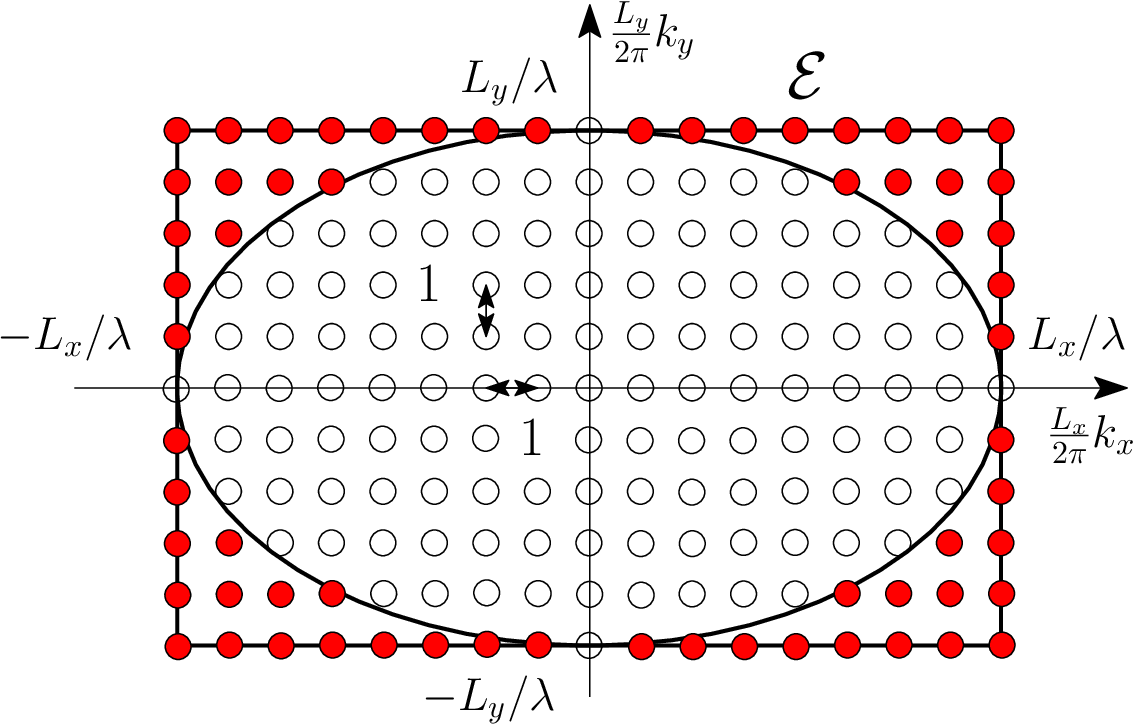} \vspace{-0.4cm}
  \caption{{The 2D lattice ellipse $\mathcal{E}$ wavenumber spectral support of $h(x,y,z)$.} \vspace{-0.7cm}
  } \label{fig:disk_lattice}
\end{figure}

\vspace{-0.3cm}
\subsection{Planar and Volumetric Arrays}

By inspection of \eqref{eq:Fourier_series_volumetric}, we can generate spatial samples of $\tilde{h}$ by using the same argument that  leads to a 2D inverse discrete Fourier transform of time-domain waveforms.
Consider a 3D parallelepiped of side lengths $L_x$, $L_y$, and $L_z<\min(L_x,L_y)$, along the three Cartesian axes, and its uniform discretization of $N = N_x N_y N_z$ points so that $N_x = \lceil L_x/\Delta_x\rceil $, $N_y = \lceil L_y/\Delta_y\rceil $, and $N_z = \lceil L_z/\Delta_z\rceil $ with spacing $\Delta_x$, $\Delta_y$, and $\Delta_z$, respectively.
In each statistical realization of $\tilde{h}$, we can generate spatial samples on the 3D uniform grid as (for an even number of points)
\begin{equation} \label{eq:FFT_volume}
\tilde{h}_{N}(x_n,y_j,z_k) \approx  \mathop{\sum\sum}_{(\ell,m)\in \mathcal{E}} \tilde{H}_{\ell m}(z_k) e^{\imagunit 2\pi \left(\frac{\ell \, n}{N_x} + \frac{m \, j}{N_y}\right)}, \quad n=-\frac{N_x}{2}, \ldots, \frac{N_x}{2}-1, \; j=-\frac{N_y}{2}, \ldots, \frac{N_y}{2}-1.
\end{equation}
which can be efficiently implemented through IFFT algorithms (e.g., \cite{FFT99}) by choosing $N$ to be an integer power of $2$. 
The Nyquist sampling condition for the spatial sampling of $\tilde{h}$ requires
\begin{equation} \label{sampling_requirement2D}
\min(\Delta_x,\Delta_y) \le \frac{\pi}{\kappa} = \frac{\lambda}{2}
\end{equation}
which is given by the fact that $\tilde{h}$ is confined to the 2D spectral support $\mathcal{D(\kappa)}$ with bandwidth less than $2\kappa=4\pi/\lambda$ (along both the $x-$ and $y-$axes).
Hence, the conventional half-wavelength antenna spacing is in general adequate to generate spatial samples of the channel, and half-wavelength arrays can be generally seen as obtained from a spatially-continuous aperture by Nyquist sampling at $\lambda/2$ intervals.
Unlike $\Delta_x$ and $\Delta_y$, the sampling interval $\Delta_{z}$ along the $z-$axis can be chosen arbitrarily. 
When IFFT algorithms \cite{FFT99} are applied, the overall complexity is of order $\mathcal{O}(N  \log(N_x N_y))$ which accounts for the cost of computing a 2D $N_xN_y-$points IFFT for each of the $N_z$ samples.
The generation of channel samples over a 2D plane aperture can be obtained from \eqref{eq:FFT_volume} by noting that $\tilde{h}(x,y)=\tilde{h}(x,y,0)$ and has a computational cost of $\mathcal{O}(N_x N_y \log(N_x N_y))$.

To summarize, Fig.~\ref{fig:Generation_new} depicts the block diagram of the Fourier plane-wave series expansion in \eqref{eq:Fourier_series_volumetric}, for a sufficiently smooth spectral factor.
To generate channel samples $\tilde{h}(x_n,y_j,z_k)$, one needs to: $\emph{i)}$ generate two 2D independent Gaussian random lattice fields $\{\tilde{H}_{\ell m}^\pm\}$, with variances $\{\sigma_{\ell m}^2\}$ as computed in Appendix~\ref{app:numerical_generation_part2}; $\emph{ii)}$ multiply them by their corresponding frequency responses $A_{h,\pm}(k_x,k_y)/(4\pi/\kappa)$ with $(\kappa_x, \kappa_y)$ being evaluated at $\left(\frac{2\pi\ell}{L_x/\lambda}, \frac{2\pi m}{L_y/\lambda}\right)$ to obtain $\{{H}_{\ell m}^\pm\}$; 
$\emph{ii)}$ apply the migration filters for every $z_k = k\Delta_{z}$ with arbitrary $\Delta_{z}$ and sum them up; $\emph{iv)}$ pass the generated lattice field $\{\tilde{H}_{\ell m}(z_k)\}$ through a 2D $N_xN_y-$point IFFT. 
Channel samples at different $z_k$ may be generated from the same $\{{H}_{\ell m}^\pm\}$. 

\begin{figure}[t!]
     \centering
     \includegraphics[width=0.8\columnwidth]{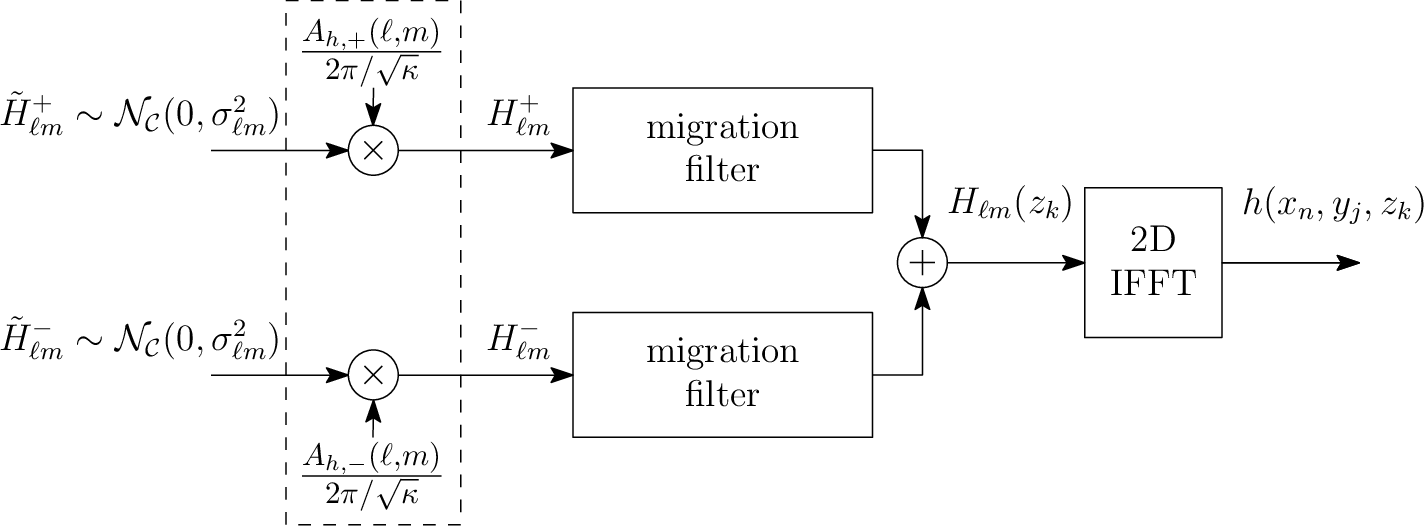} \vspace{-0.4cm}
  \caption{{Block diagram of the numerical generation procedure of 3D spatial channel samples.} \vspace{-0.7cm}
  } \label{fig:Generation_new}
\end{figure}

\vspace{-0.3cm}
\subsection{Linear Arrays}

Given a 1D uniform spatial grid of $N = \lceil L_x/\Delta_x \rceil$ points with spacing $\Delta_x$, we can generate samples of $\tilde{h}$ as
\begin{equation} \label{eq:FFT_linear}
\tilde{h}_{N}(x_n) \approx \sum_{\ell = -L_x/\lambda}^{L_x/\lambda-1} \tilde{H}_{\ell} \, e^{\imagunit  2\pi \ell n/N}, \quad n=-\frac{N}{2}, \ldots, \frac{N}{2}-1
\end{equation}
where $\Delta_x$ is chosen according to the Nyquist sampling condition for a $2\kappa-$bandlimited process
\begin{equation} \label{sampling_requirement}
\Delta_x \le \frac{\lambda}{2}.
\end{equation}
If IFFT is used \cite{FFT99}, the overall complexity of the channel generation procedure is $\mathcal{O}(N \log(N))$.
According to the procedure described above and summarized in Fig.~\ref{fig:Generation_new}, we can generate 1D channel samples by passing \eqref{eq:FFT_linear} through two linear space-invariant filters with wavenumber responses driven by the spectral factors $A_{h,\pm}(k_x)$.

\begin{figure}[t!]
       \centering
      \includegraphics[width=.7\columnwidth]{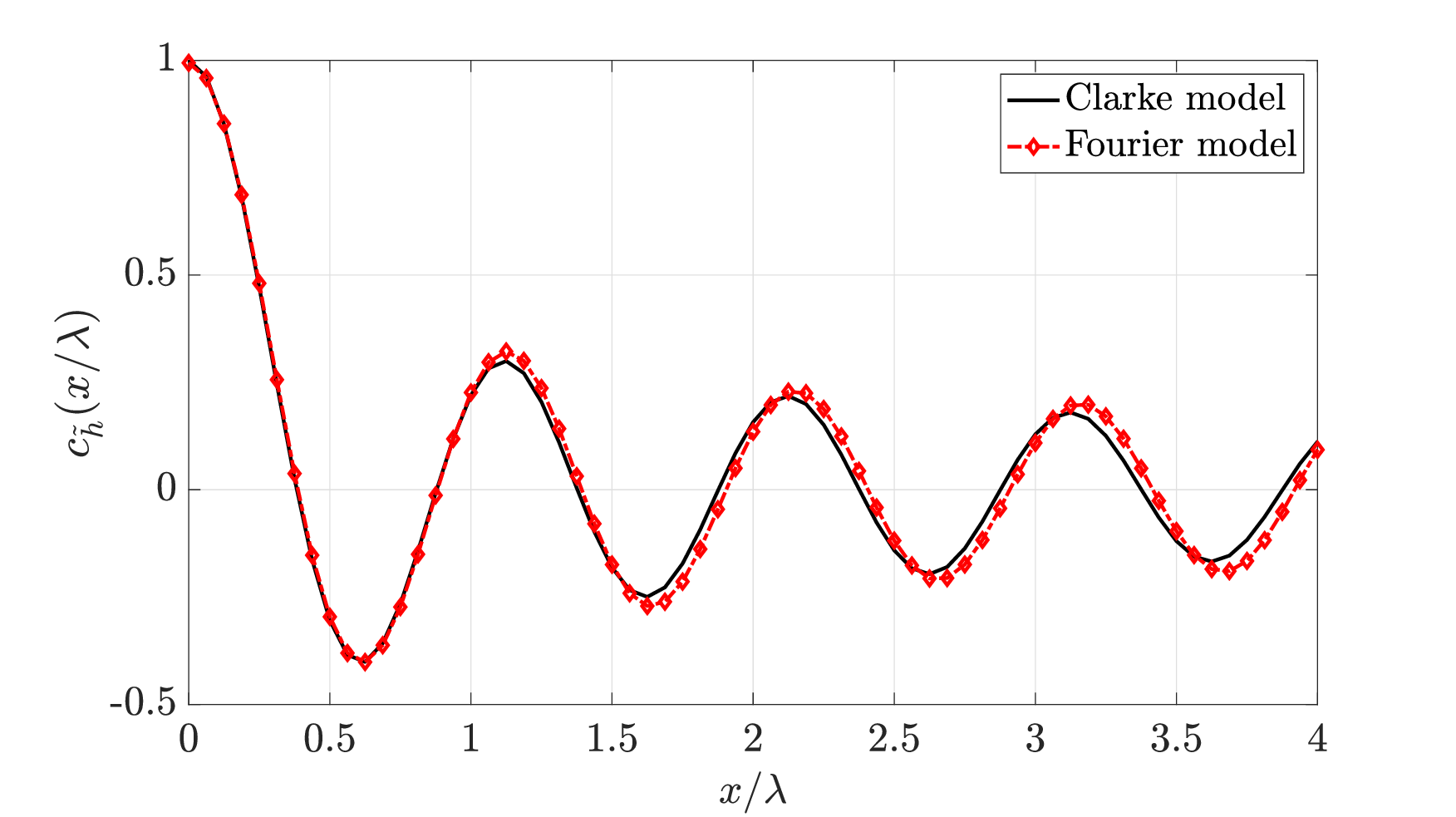} \vspace{-0.4cm}
     \caption{\vspace{-0cm}1D autocorrelation of $h(x)$ as a function of $x/\lambda \in [0, L_x/4]$ with $L_x=16\lambda$, and $\Delta_x= \lambda/16$.      
     \vspace{-0.5cm}}\label{fig:linear_aperture}
\end{figure}

 \vspace{-0.3cm}
\subsection{Numerical Validation} \label{sec:Simulation}

Numerical results are now used to validate the accuracy of the analytical framework developed above. 
We focus on the isotropic propagation scenario since it is the key to generate any non-isotropic channel.
The channel samples are generated as shown above for any $N$-dimensional uniform spatial grid and then collected into a random vector $\tilde{\vect{h}}_N \in\Complex^{N}$. 
The accuracy of the proposed method is compared to the state-of-the-art model of spatially-stationary random field channels, i.e., the discrete Karhunen-Lo{\`e}ve representation $\tilde{\vect{h}}_N = \vect{C}_{\tilde{h}}^{1/2} \vect{e}$
where $\vect{e} \sim \CN(\vect{0},\vect{I}_N)$, and $\vect{C}_{\tilde{h}}\in\Complex^{N\times N}$ is the spatial correlation matrix.
This matrix is computed by sampling Clarke's autocorrelation function \cite{Clarke,Aulin79}, which for a linear aperture is given by $c_h(r) = J_0(2\pi r/\lambda)$ (see Lemma~\ref{ACF_Jakes2D}), whereas it is $c_h(r) = \sinc(2r/\lambda)$ (see Lemma~\ref{ACF_Jakes}) with volumetric and planar apertures.
In general, $\vect{C}_{\tilde{h}}$ is semidefinite positive and has a symmetric block-Toeplitz structure with entries $[\vect{C}_h]_{nm} = c_h(r_{nm})$ with $n,m=1,\ldots, N$, where $r_{nm}$ is the distance between the $n-$th and $m-$th grid points. By choosing a uniform spacing $\Delta$ along the $x-$ and $y-$axes, $\vect{C}_{\tilde{h}}$ becomes of symmetric Toeplitz structure and is therefore fully characterized by its first row (or column). Hence, we compare these two methods by plotting the first row of their spatial correlation matrices. 

We begin by considering a linear aperture with $L_x=16\lambda$. Fig.~\ref{fig:linear_aperture} illustrates the 1D autocorrelation function of the numerically generated samples $\tilde{h}(x_n)$ with spatial sampling ${\Delta = \lambda/16}$. 
 As it is seen, the empirical autocorrelation function matches well its closed form, which is known a priori and given by ${c_{\tilde{h}}(r_n) = J_0\left(\frac{2\pi r_n}{\lambda}\right)}$ with $r_n=x_n$.
Fig.~\ref{fig:planar_aperture} plots the 2D autocorrelation function of $\tilde{h}_{N}(x_n,y_j,0)$ over a rectangular grid on the plane $z_k=0$ of side lengths ${L_x=L_y=16\lambda}$ with uniform spacing ${\Delta = \lambda/4}$.
 Similar conclusions as for 1D apertures hold. Finally, Fig.~\ref{fig:volume_aperture} validates the effect of migration filters to obtain the 2D autocorrelation function of $\tilde{h}_{N}(x_n,y_j,z_k)$ over the same rectangular grid on the plane $z_k=\lambda/2$. These numerical results validate the accuracy of the developed numerical procedure for both the 2D and 3D propagation models and the applicability of the Fourier plane-wave series expansion of the channel to model the field over compact rectangular arrays of practical size.

\begin{figure}[t!]
       \centering
      \includegraphics[width=.7\columnwidth]{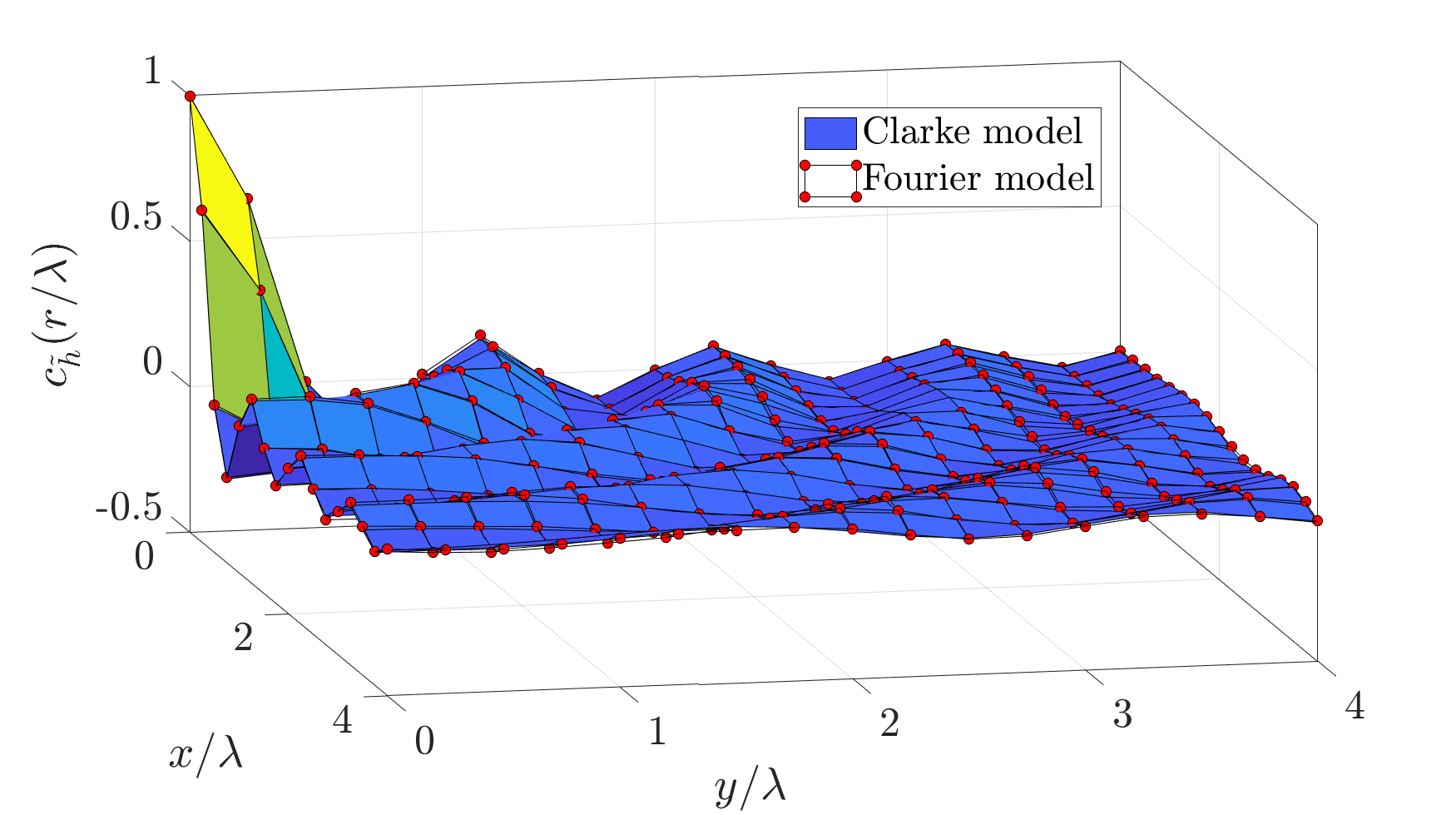} \vspace{-0.4cm}
     \caption{\vspace{-0cm}2D autocorrelation of $h(x,y,0)$ as a function of $x/\lambda \in [0, L_x/4]$ and $y/\lambda \in [0, L_y/4]$ with $L_x=L_y=16\lambda$, and $\Delta_x=\Delta_y= \lambda/4$.      
     \vspace{-0.5cm}}\label{fig:planar_aperture}
\end{figure}

 \begin{figure}[t!]
       \centering
      \includegraphics[width=.7\columnwidth]{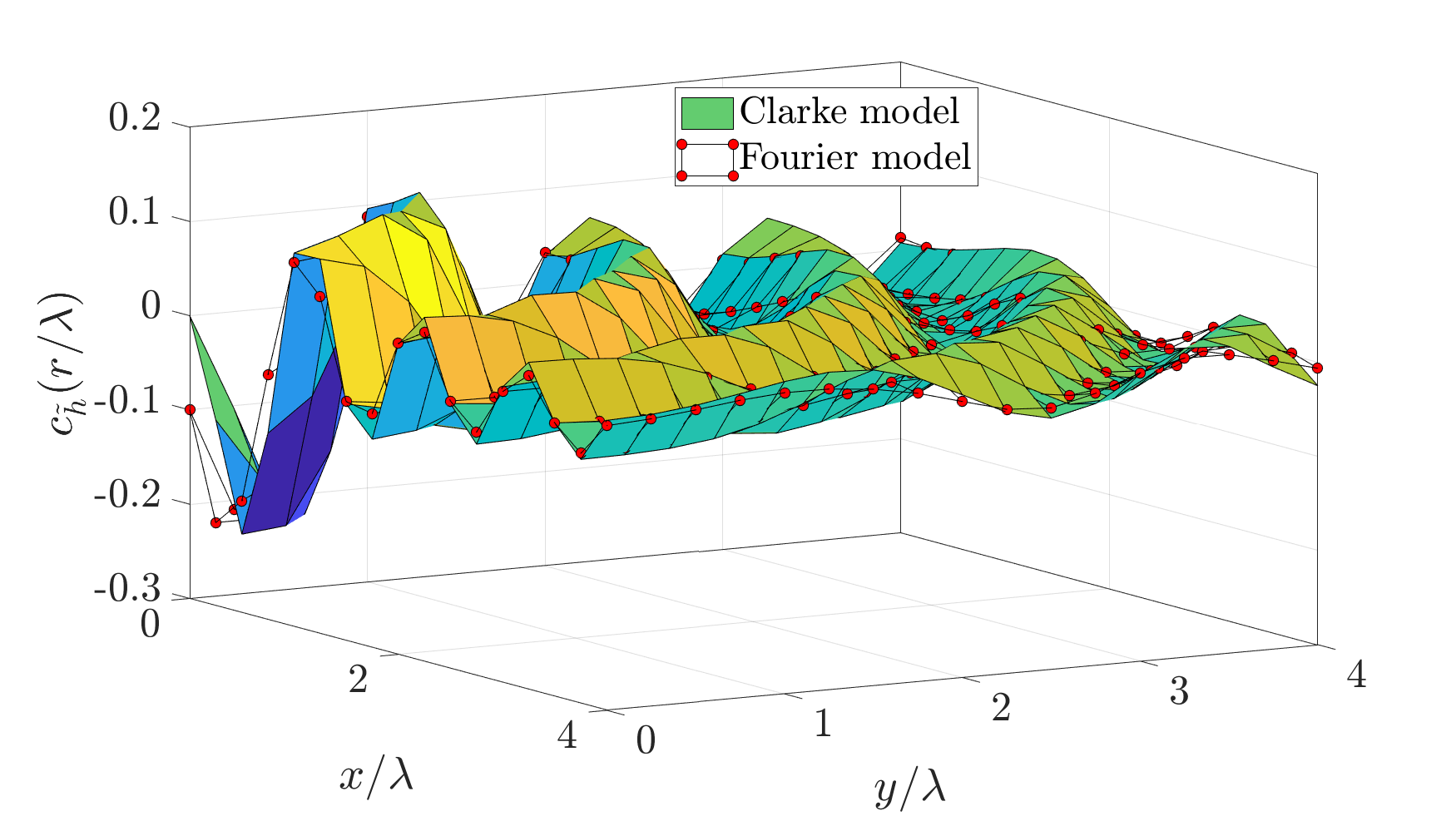} \vspace{-0.4cm}
     \caption{\vspace{-0cm}{2D autocorrelation of $h(x,y,\lambda/2)$ as a function of $x/\lambda \in [0, L_x/4]$ and $y/\lambda \in [0, L_y/4]$ with $L_x=L_y=16\lambda$, and $\Delta_x=\Delta_y= \lambda/4$.}      
     \vspace{-0.5cm}}\label{fig:volume_aperture}
\end{figure}

\section{Conclusions and Outlook}\label{sec:conclusions}
Holographic MIMO arrays, thought of as spatially-constrained MIMO arrays with a massive number of antennas $N$, are considered as a possible solution to approach the practical Massive MIMO limit $N\to \infty$. To obtain a physically-meaningful stochastic description of non-isotropic radio waves propagation in the far-field, we 
modeled the small-scale fading as a zero-mean, spatially-stationary, complex-Gaussian and scalar random field, that satisfies the Helmholtz equation. This modeling led directly to the only physically-meaningful power spectral density, which is given by the product between a Dirac delta function in the three wavenumber components, and a non-negative amplitude term that defines directional weighting.
The structure of the power spectral density provided an exact 2D Fourier plane-wave spectral representation for the small-scale fading with a singularly-integrable and bandlimited spectrum. 
Such a representation led to a 2D Fourier plane-wave series expansion for the field over spatially-constrained compact spaces, which, suitably discretized, provided an accurate and computationally-efficient numerical procedure to generate small-scale fading samples of Holographic MIMO arrays. Numerical results were used to validate this procedure with compact arrays of practical size.

We anticipate that the proposed analytical framework will be a valuable tool for the theoretical analysis of Holographic MIMO systems in the presence of frequency-flat fading.
 For example, in \cite{Pizzo2020}, it is directly used to determine the upper limit to the available degrees of freedom (DoF). It turns out that the DoF per m of a linear aperture deployment are asymptotically (as the aperture size increases) limited to $2/\lambda$ with $\lambda$ being the wavelength. For a planar deployment, the DoF per m$^2$ are limited to $\pi/\lambda^2$. The expansion of a planar aperture into a volume aperture asymptotically yields only a two-fold increase in the available DoF. The result is in agreement with previous works (e.g., \cite{Rusek2018,Franceschetti, PoonDoF}), and imposes a limit on the number of parallel channels that can be established on a communication link. Our treatment can be applied to the analysis of any multi-user communication system in which all the terminals have single antennas and are sufficiently separated in space, i.e., the spatial correlation occurs only among the service antennas (e.g. multi-user Holographic MIMO). The developed framework can also be extended to the case of vector electromagnetic random fields wherein each component of the field is a function of six Cartesian coordinates denoting the spatial positions of transmitter and receiver \cite{Marzetta2018}. This is instrumental for computing the capacity of point-to-point Holographic MIMO systems.

Finally, we observe that our treatment provides an exact representation for the small-scale fading, which is valid at any frequency range. In particular, the accuracy of the proposed model increases as the array size becomes larger compared to the wavelength. This observation makes the proposed model appealing to conduct both theoretical and numerical analysis of high-frequency communication systems operating at mmWave and THz bands where Holographic MIMO arrays are promising, and small-scale fading modeling is a fundamental challenge \cite{mmWave,Thz}.

\appendices
\vspace{-0.3cm}

%

\section*{Appendix I}
\section*{Reviewing the Fourier Spectral Representation} \label{app:spectral_representation}
Every deterministic signal of finite power can be represented either in the time domain as a waveform or in the frequency domain as a spectrum. The mapping between these two domains is the Fourier transform.
Similarly, for every zero-mean, second-order, stationary random process $y(t)$ with power spectral density $S_y(\omega)$ the {Fourier spectral representation} \cite[Eq.~(223)]{VanTreesBook}:
\begin{equation} \label{eq:Spectral_representation}
{y(t) = \frac{1}{\sqrt{2\pi}} \int_{-\infty}^{\infty} e^{\imagunit \omega t} \, dY(\omega)}
\end{equation}
provides a frequency-domain description. The integral must be interpreted in the stochastic mean-square sense and $Y(\omega)$ is the integrated Fourier transform of $y(t)$; that is, it is a Wiener process
such that in differential form $\Ex\{dY(\omega) dY^*(\omega^\prime)\} = 0 $ for $\omega\ne \omega^\prime$ and
\begin{align}
\Ex\{|dY(\omega)|^2\} = \frac{d\omega}{2\pi}  S_y(\omega)\label{eq:diff_equation_1}
\end{align}
where $dY(\omega)$ is the increment integrated Fourier transform \cite[Sec.~3.6]{VanTreesBook}.
To rewrite \eqref{eq:Spectral_representation} as a function of $S_y(\omega)$, we can proceed as follows.
Take a stationary random process $x(t)$ with power spectral density $S_x(\omega)$ and pass it through a linear time-invariant system with arbitrary frequency response $F(\omega)$. The output $y(t)$ is a random process such that \cite[Eq.~(233)]{VanTreesBook}
\begin{equation}\label{eq:Spectral_LTIfilter}
dY(\omega) = dX(\omega) F(\omega)
\end{equation}
from which it follows that $\Ex\{|dY(\omega)|^2\} = |F(\omega)|^2 \Ex\{|dX(\omega)|^2\}$
with $\Ex\{|dX(\omega)|^2\}   =  \frac{d\omega}{2\pi}  S_x(\omega)$. Combing this result with \eqref{eq:diff_equation_1} yields $S_y(\omega) = S_x(\omega) |F(\omega)|^2$ from which
\begin{equation}  \label{eq:Spectral_LTIfilter_PSD}
F(\omega) = \sqrt{S_y(\omega)/S_x(\omega)}.
\end{equation}
From \eqref{eq:Spectral_representation}, by plugging \eqref{eq:Spectral_LTIfilter_PSD} into \eqref{eq:Spectral_LTIfilter}, it follows that we can generate $y(t)$ with a given power spectrum $S_y(\omega)$ by passing a stationary white-noise random process
$x(t)$ (i.e., with $S_x(\omega)=1$ for $\omega\in\Real$) through a linear time-invariant filter with frequency response $F(\omega) = \sqrt{S_y(\omega)}$. This provides us with the linear-system form of the Fourier spectral representation
\begin{equation} \label{eq:Spectral_representation_PSD_Lebesgue}
y(t)  = \frac{1}{\sqrt{2\pi}} \int_{-\infty}^{\infty} \sqrt{S_y(\omega)} e^{\imagunit\omega t} \, dX(\omega).
\end{equation}
The Riemann integral form of \eqref{eq:Spectral_representation_PSD_Lebesgue} reads as
\begin{equation} \label{eq:Spectral_representation_PSD}
y(t)  = \frac{1}{\sqrt{2\pi}} \int_{-\infty}^{\infty} \sqrt{S_y(\omega)} W(\omega) e^{\imagunit\omega t} \, d\omega
\end{equation}
where $W(\omega)$ is a white-noise Gaussian process with unit spectrum, which can be seen as a superposition of an uncountably-infinite number of harmonics having statistically-independent Gaussian-distributed random coefficients.
For any second-order random process (i.e, with $\int_{-\infty}^{\infty} S_y(\omega) d\omega < \infty$), the above spectral representation converges in the mean-square sense and finds its justification through the linear functional of a white-noise random process $W(\omega)$ with a complex-valued square-integrable function \cite[Sec.~7.4]{GallagerBook}. 
In fact, $W(\omega)$ can be viewed as a generalized random process with $c_W(\omega)=\Ex\{W(\omega + \omega^\prime) W^*(\omega^\prime)\}=\delta(\omega)$ \cite[Sec.~7.7]{GallagerBook} in the same manner the Dirac delta function can be viewed as a generalized function, or distributions \cite{Johnson}. 
Hence, a white-noise process can be regarded as the stochastic counterpart to the Dirac delta function, which provides us with a system-theoretic interpretation of \eqref{eq:Spectral_representation_PSD}.
In particular, the generation of a random process with a given power spectral density is analogous to the generation of a deterministic signal with a given Fourier transform; see Fig.~\ref{fig:LTI_2}.


\begin{figure}[t!]
    \centering
     \includegraphics[width=.9\columnwidth]{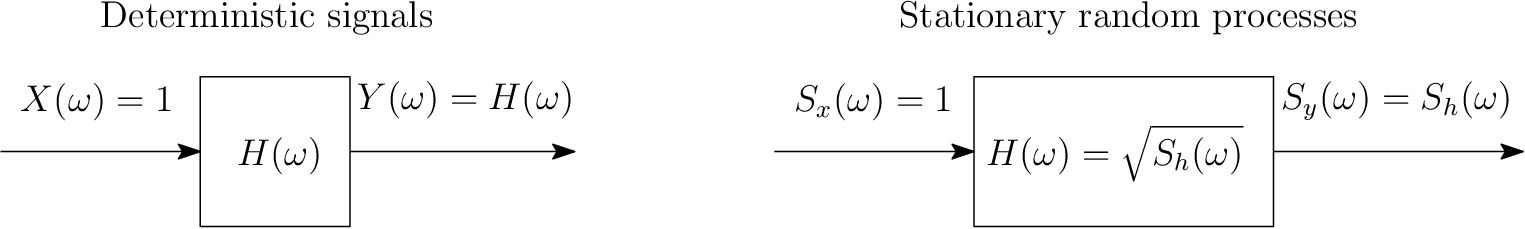}
     \caption{\vspace{-0cm}Analogy between Dirac delta function and white-noise process in linear systems.
     }\label{fig:LTI_2}\vspace{-0.8cm}
\end{figure}

\section*{Appendix II}
\section*{Power Normalization} \label{app:power_normalization}
Assume that $\tilde{h}(x,y,z)$ has unit power, i.e., $1/(2\pi)^3 \iiint_{-\infty}^{\infty} S_{\tilde{h}}(k_x,k_y,k_z) \, dk_x dk_y dk_z =1$.
By substituting \eqref{eq:spectral_factor3D} into \eqref{eq:PSD_function}, we obtain the following condition
\begin{equation} \label{eq:normalized_spectral_factor}
A_h^2(\kappa) =  \frac{(2\pi)^3}{\iiint_{-\infty}^{\infty}  \delta(k_x^2 + k_y^2 + k_z^2 - \kappa^2) \, dk_x dk_y dk_z}.
\end{equation}
The above integral can be solved by a change of integration variables to spherical coordinates
\begin{align} \label{eq:int_spherical}
 \int_{0}^{2\pi} \int_{0}^{\pi} \int_{0}^{\infty}  \delta(k_r^2 - \kappa^2) k_r^2 \sin(k_\theta) \, dk_\phi dk_\theta dk_r =  4\pi \int_{0}^{\infty}  \delta(k_r^2 - \kappa^2) k_r^2 \, dk_r 
\end{align}
where we have used $\int_{0}^{2\pi} \int_{0}^{\pi} \sin(k_\theta) \, dk_\phi dk_\theta = 4\pi$. We now observe that \cite[Eq.~181.a]{ArfkenBook}
\begin{equation} \label{eq:Dirac_delta_g_multivariate}
\delta(k_r^2 - \kappa^2) = \frac{\delta(k_r - \kappa) + \delta(k_r + \kappa)}{2\kappa}
\end{equation}
which substituted into \eqref{eq:int_spherical} yields $2\pi \kappa$, where we picked up the positive zero only.
The normalizing spectral factor is thus given by ${A_h^2(\kappa) =   4\pi^2/\kappa}$.
The spectral factor for an isotropic 2D channel can be obtained by following the same arguments above by using polar coordinates instead.


\section*{Appendix III}
\section*{Proof of Lemma~\ref{ACF_Jakes}} \label{app:autocorrelation_isotropic}

The substitution of \eqref{eq:spectral_density3D} into \eqref{eq:ACF_W} yields the autocorrelation function of an isotropic channel
 \begin{equation} \label{eq:c_iso}
c_{\tilde{h}}(x,y,z)  =  \frac{1}{4\pi \kappa} \iiint_{-\infty}^{\infty} \delta(k_x^2 + k_y^2 + k_z^2 - \kappa^2)  e^{\imagunit (k_x x + k_y y + k_z z)}\, dk_x dk_y dk_z.
\end{equation}
The above integral is a 3D Fourier inverse transform of a spherically symmetric function $\delta(k_x^2 + k_y^2 + k_z^2 - \kappa^2)$. 
Thus, by aligning the spatial vector to one of the axis, say the $z-$axis, such that $(x,y,z) = (0,0,R)$ with $R=\sqrt{x^2 + y^2 + z^2}$ we obtain
 \begin{equation} \label{eq:c_iso2}
c_{\tilde{h}}(x,y,z)   =  \frac{1}{2\pi \kappa} \iiint_{-\infty}^{\infty} \delta(k_x^2 + k_y^2 + k_z^2 - \kappa^2)  e^{\imagunit k_z R}\, dk_x dk_y dk_z.
\end{equation}
The change of integration variables from Cartesian to spherical ${(k_x,k_y,k_z) \to (k_r,k_\theta,k_\phi)}$ yields
\begin{equation} \label{eq:c_iso3} 
c_{\tilde{h}}(x,y,z)  = \frac{1}{2\pi \kappa}  \int_{0}^{\infty}  \int_{0}^{\pi}  \int_{0}^{2\pi}  \delta(k_r^2 - \kappa^2) e^{\imagunit k_r R \cos(k_\theta)} k_r^2 \sin(k_\theta) \,  dk_r dk_\theta dk_\phi
\end{equation}
which by using \eqref{eq:Dirac_delta_g_multivariate} and exploiting the angular symmetry over $k_\phi$ leads to
\begin{equation} \label{eq:c_iso4} 
c_{\tilde{h}}(x,y,z) = \frac{1}{2}   \int_{0}^{\pi}  e^{\imagunit \kappa R \cos(k_\theta)} \sin(k_\theta) \,  dk_\theta \mathop{=}^{(a)} j_0(\kappa R)
\end{equation}
where $(a)$ follows from Poisson's integral \cite[Eq.~10.1.14]{AbramowitzStegun} 
and $j_0(x)$ is the spherical Bessel function of the first kind and order $0$, defined as $j_0(x) = \sin(x)/x$.
The autocorrelation function of an isotropic 2D channel
 follows similarly from polar coordinates and the integral representation of Bessel's functions of first kind $\frac{1}{\pi}\int_{0}^\pi  e^{\imagunit z\cos(k_\theta)}  \, dk_\theta = J_0(z)$ \cite[Eq.~(9.1.21)]{AbramowitzStegun}.

\vspace{-0.3cm}
\section*{Appendix IV} \label{app:numerical_generation}

\subsection{Fourier spectral series expansion} \label{app:numerical_generation_part1}

Consider a zero-mean, second-order, stationary Gaussian random process $y(t)$ having a bandlimited \emph{singularly-integrable} spectrum $S_y(\omega)$ with $\omega\in[-\Omega, \Omega]$ of angular-frequency bandwidth $\Omega < \infty$ and defined over a  time interval $t\in(-\infty,\infty)$.
The, its frequency-domain description is given by the Fourier spectral representation in \eqref{eq:Spectral_representation_PSD}.
Let us now observe $y$ over a \emph{large, but finite}, time interval $t\in[0, T]$ of duration $T<\infty$.
We aim to provide a discrete-frequency approximation for the Fourier spectral representation in \eqref{eq:Spectral_representation_PSD} over $t\in[0, T]$. 
We start by partitioning the integration interval uniformly with frequency spacing ${\Delta_\omega=2\pi/T}$
\begin{equation} \label{eq:y_partition}
y(t)  \mathop = \frac{1}{\sqrt{2\pi}} \sum_{\ell=-\Omega T/\pi}^{\Omega T/\pi-1} \left(\int_{ 2\pi \ell/T}^{2\pi (\ell + 1)/T} \sqrt{S_y(\omega)} W(\omega) e^{\imagunit\omega t} \, d\omega \right), \quad t\in[0, T].
\end{equation}
Applying the first mean-value theorem \cite[Ch.~3]{IntegralBook} over each interval $ \omega\in[ \frac{2\pi\ell}{T}, \frac{2\pi (\ell + 1)}{T}]$ yields
\begin{equation} \label{eq:y_partition2}
y(t)  \approx \frac{1}{\sqrt{2\pi}} \sum_{\ell=-\Omega T/\pi}^{\Omega T/\pi-1}  \left( \int_{ 2\pi \ell/T}^{2\pi (\ell + 1)/T} \sqrt{S_y(\omega)} W(\omega)  \, d\omega \right) \, e^{\imagunit 2\pi (\ell + 1/2) t/T}, \quad t\in[0, T]
\end{equation}
where the approximation error becomes negligible as ${\Delta_\omega/\Omega \to 0}$ (i.e, $\Omega T\to \infty$).
The integral between brackets is a linear functional of a white Gaussian noise process $W(\omega)$ \cite[Ch.~7.4]{GallagerBook}
\begin{equation} \label{Y_l}
Y_\ell = \frac{1}{\sqrt{2\pi}} \int_{ 2\pi \ell/T}^{2\pi (\ell + 1)/T} \sqrt{S_y(\omega)} W(\omega)  \, d\omega= \int_{-\infty}^{\infty} g_\ell(\omega) W(\omega) \, d\omega 
\end{equation}
with a real-valued square-integrable ($\mathcal{L}_2$) function
 \begin{equation} \label{eq:orthogonal_functions}
g_\ell(\omega) = 
\begin{cases}
\sqrt{S_y(\omega)/2\pi} & \omega\in[ \frac{2\pi\ell}{T}, \frac{2\pi (\ell + 1)}{T}] \\
0 & \mathrm{elsewhere}.
\end{cases}
\end{equation} 
We notice that $\{g_\ell(\omega)\}$ with $\ell\in\{-\Omega T/\pi, \ldots, \Omega T/\pi-1\}$ are such that $\int_{-\infty}^{\infty} g_\ell(\omega) g_m(\omega)\, d\omega = \sigma_\ell^2 \delta_{m-\ell}$ and each function has energy
\begin{equation} \label{energy_l}
\sigma_\ell^2 = \int_{ 2\pi \ell/T}^{2\pi (\ell + 1)/T} S_y(\omega) \, \frac{d\omega}{2\pi} < \infty
\end{equation}
which is finite for every second-order random process.
Thus $\{g_\ell(\omega)\}$ form a set of orthogonal functions, or better, an orthogonal basis for the space of real-valued $\mathcal{L}_2$ functions.
As a consequence, we can interpret \eqref{Y_l} as the $\ell-$th coordinate of an orthonormal series expansion of $W$ over the basis $\{g_\ell(\omega)\}$.
Now, the expansion of any white Gaussian noise process over an arbitrary orthogonal $\mathcal{L}_2$ basis of functions produces a sequence ${Y_\ell \sim \CN(0, \sigma_\ell^2)}$ of independent zero-mean, circularly-symmetric Gaussian random variables with finite variances $\sigma_\ell^2$ \cite[Ch.~7.7]{GallagerBook}.
Thus, we may approximate \eqref{eq:Spectral_representation_PSD} for a bandlimited spectra $S_y$ as 
\begin{equation} \label{eq:DTFT}
y(t)  \approx  \sum_{\ell=-\Omega T/\pi}^{\Omega T/\pi-1} Y_\ell \, e^{\imagunit 2\pi (\ell+1/2) t/T} \mathop \sim \limits^{(a)} \sum_{\ell=-\Omega T/\pi}^{\Omega T/\pi-1} Y_\ell \, e^{\imagunit 2\pi \ell t/T}, \quad t\in[0, T]
\end{equation}
where $(a)$ holds in the statistical distribution sense and it is due the phase-invariance of circularly-symmetric Gaussian random variables.
The approximation error becomes negligible as $\Omega T \to \infty$.
The above formula provides an orthonormal decomposition of $y$ taking the form of a bandlimited Fourier spectral series expansion having a \emph{countably-finite} number of statistically-independent Gaussian-distributed coefficients.
Hence, it provides an approximation of $y$ over $t\in[0,T]$ through a \emph{periodic} stationary random process over its fundamental period $T=2\pi/\Delta_\omega$.
The convergence of \eqref{eq:DTFT} in mean square may be proven by recurring to the Parseval's theorem and use \eqref{energy_l}
\begin{equation}
\Ex\{|y(t)|^2\} =  \sum_{\ell=-\Omega T/\pi}^{\Omega T/\pi-1} \sigma_\ell^2 \mathop = \int_{-\Omega}^{\Omega} S_y(\omega) \, \frac{d\omega}{2\pi}  < \infty.
\end{equation}
Similar to the Fourier series of a time-domain waveform, \eqref{eq:DTFT} provides a periodic output. However, the Fourier coefficients $Y_\ell$ of $y(t)$ cannot be simply obtained by frequency sampling its spectra at $\ell \Delta_\omega=2\pi \ell/T$. This is better explained next.

\vspace{-0.3cm}
\subsection{Connection to the Karhunen-Lo{\`e}ve series expansion}

The Fourier series expansion \eqref{eq:DTFT} is reminiscent of the famous Karhunen-Lo{\`e}ve series expansion $y(t) = \sum_\ell c_\ell \varphi_\ell(t)$ with $t\in[0,T]$ where $c_\ell$ and $\varphi_\ell$ are the eigenvalues and eigenfunctions of such expansion.
This is the continuous analog of the decomposition of a random vector into its principal components (i.e., eigenvalue decomposition) \cite{VanTreesBook}.
In particular, for a bandlimited random process of bandwidth $\Omega$, \eqref{eq:DTFT} can be interpreted as the asymptotic version of the Karhunen-Lo{\`e}ve series expansion \cite[Sec.~3.4.6]{VanTreesBook}.
In fact, as the observation interval $T$ grows large (i.e., $\Omega T\gg 1$) the eigenvalues' power $\Ex\{|c_\ell|^2\}$ approach the power spectral density $S_y$, and the eigenfunctions $\varphi_\ell$ become harmonics.
Thus, the Fourier coefficients $\{Y_\ell\}$ and basis $\{e^{\imagunit 2\pi\ell t/T}\}$ respectively assume the meaning of Karhunen-Lo{\`e}ve's eigenvalues and eigenfunctions.

We finally notice that the direct application of the Karhunen-Lo{\`e}ve expansion to the modeling of $y$ would lead to a divergent series due to the evaluation of the singular spectrum $S_y$ on its boundary. This is analogous to the problem encountered in the use of the Fourier spectral representation in Section~\ref{sec:Spectral}.
Asymptotically as $\Omega T \to \infty$, the Fourier coefficients become a continuum and both the Fourier series and Karhunen-Lo{\`e}ve expansion tend to the Fourier spectral representation in \eqref{eq:Spectral_representation_PSD}.

\vspace{-0.3cm}
\subsection{Computation of variances of Fourier coefficients} \label{app:numerical_generation_part2}

The variance $\sigma_\ell^2$ of the $\ell-$th Fourier random coefficient $H_\ell$ in \eqref{eq:Fourier_series_linear} is computed from \eqref{energy_l} by replacing the time-frequency mapping with its space-wavenumber counterpart:
\begin{equation} \label{variance_linear_aperture}
\sigma_\ell^2 = \int_{ 2\pi \ell/L_x}^{2\pi (\ell + 1)/L_x} S_{\tilde{h}}(k_x) \, \frac{dk_x}{2\pi}  = \int_{ 2\pi \ell/L_x}^{2\pi (\ell + 1)/L_x} \frac{1}{\sqrt{\kappa^2 - k_x^2}} \, \frac{dk_x}{2\pi} =  \frac{1}{2\pi} \int_{\ell \lambda/L_x}^{(\ell+1) \lambda/L_x} \frac{1}{\sqrt{1 - k_x^2}} \, dk_x
\end{equation}
where we substitute the 1D power spectrum $S_{\tilde{h}}$ in \eqref{eq:Spectral_Factor_projection_2d} and use $\kappa=2\pi/\lambda$.
Since the integrand is symmetric with respect the origin, the variances corresponding to the negative indexes can be easily found by symmetry $\sigma_{-\ell-1}^2 = \sigma_{\ell}^2$ for $\ell=0, 1, \ldots, L_x/\lambda-1$.
By applying, e.g., the change of integration variable $k_x = \sin u$ we obtain
\begin{equation} \label{sigma_elle_positive3}
\sigma_\ell^2 = \frac{1}{2\pi} \left( \mathrm{arcsin}\left((\ell+1) \frac{\lambda}{L_x}\right)  -  \mathrm{arcsin}\left(\ell \frac{\lambda}{L_x}\right) \right), \quad \ell=0, 1, \ldots, L_x/\lambda-1.
\end{equation}
Similarly, from \eqref{energy_l} the variance $\sigma_{\ell,m}^2$ of the $(\ell,m)-$th Fourier random coefficient $H_{\ell m}$ in \eqref{eq:Fourier_series_planar} is
\begin{align}  \notag
\sigma_{\ell,m}^2  & =  \int_{2\pi \ell/L_x}^{2\pi(\ell+1)/L_x} \int_{2\pi m/L_y}^{2 \pi (m+1)/L_y} S_{\tilde{h}}(k_x,k_y) \, \frac{dk_x}{2\pi} \frac{dk_y}{2\pi},  \quad (\ell,m)\in \mathcal{E}
\\& \notag
 =  \int_{2\pi \ell/L_x}^{2\pi(\ell+1)/L_x} \int_{2\pi m/L_y}^{2 \pi (m+1)/L_y} \frac{\pi}{\kappa} \frac{\mathbbm{1}_{\mathcal{D}(\kappa)}(k_x,k_y)}{\sqrt{\kappa^2 - k_x^2 - k_y^2}} \, \frac{dk_x}{2\pi} \frac{dk_y}{2\pi},  \quad (\ell,m)\in \mathcal{E}
\\& \label{sigma_elle_p_positive}
 =  \frac{1}{4\pi} \int_{\ell \lambda/L_x}^{(\ell+1) \lambda/L_x} \int_{m \lambda/L_y}^{(m+1) \lambda/L_y} \frac{\mathbbm{1}_{\mathcal{D}(1)}(k_x,k_y)}{\sqrt{1 - k_x^2 - k_y^2}} \, dk_x dk_y,  \quad (\ell,m)\in \mathcal{E}
\end{align}
where we substitute the 2D power spectrum $S_{\tilde{h}}$ in \eqref{eq:Spectral_Factor_projection_3d_isotropic} and use $\kappa=2\pi/\lambda$. The set $\mathcal{E}$ is the 2D lattice ellipse as defined in \eqref{ellipse}.
Due to the rotational symmetry of the integrand we focus on the first wavenumber quadrant only, that is, $\ell=0, 1, \ldots, L_x/\lambda-1$ and $m=0, 1, \ldots, L_y/\lambda-1$ from which we can derive the variances in all the other quadrants.  
By change of integration variable to polar coordinates $(k_x,k_y) = (k_r \cos k_\phi, k_r \sin k_\phi)$ with $k_r\in[0,1]$ and $k_\phi\in[0,\pi/2)$, for the case $\ell\ge m$ we obtain
\begin{align} \notag
\sigma_{\ell,m}^2  & =  \frac{1}{4\pi} \Bigg(\int_{k_{\phi,1}}^{k_{\phi,2}} \int_{\min\left(1,\frac{m \lambda}{L_y \sin k_\phi}\right)}^{\min\left(1,\frac{(\ell+1)\lambda}{L_x \cos k_\phi}\right)} \frac{k_r }{\sqrt{1 - k_r^2}} \, dk_r dk_\phi 
+ \int_{k_{\phi,2}}^{k_{\phi,3}} \int_{\min\left(1,\frac{\ell \lambda}{L_x \cos k_\phi}\right)}^{\min\left(1,\frac{(\ell+1) \lambda}{L_x \cos k_\phi}\right)} \frac{k_r }{\sqrt{1 - k_r^2}} \,  dk_r dk_\phi \\&   \label{sigma_elle_p2} 
+ \int_{k_{\phi,3}}^{k_{\phi,4}} \int_{\min\left(1,\frac{\ell \lambda}{L_x \cos k_\phi}\right)}^{\min\left(1,\frac{(m+1) \lambda}{L_y \sin k_\phi}\right)} \frac{k_r }{\sqrt{1 - k_r^2}} \Bigg) \,  dk_r dk_\phi, \quad \ell\ge m
\end{align}
where $k_{\phi,1} = \arctan\left(\frac{m L_x}{(\ell+1) L_y}\right)$, $k_{\phi,2} = \arctan\left(\frac{m L_x}{\ell L_y}\right)$, $k_{\phi,3} = \arctan\left(\frac{(m+1) L_x}{(\ell+1) L_y}\right)$, and $k_{\phi,4} = \arctan\left(\frac{(m+1) L_x}{\ell L_y}\right)$ such that $k_{\phi,1} \le k_{\phi,2} \le k_{\phi,3}  \le k_{\phi,4}$.
Now, the integration of \eqref{sigma_elle_p2} over the radial wavenumber component yields to
\begin{align} \notag
&  \sigma_{\ell,m}^2  =  \frac{1}{4\pi} \Bigg(\int_{k_{\phi,1}}^{k_{\phi,2}} \sqrt{1 - \min\left(1, \frac{1}{a_{11}^2 \sin^2 k_\phi}\right)} \, dk_\phi  - \int_{k_{\phi,1}}^{k_{\phi,2}} \sqrt{1 - \min\left(1, \frac{1}{a_{12}^2 \cos^2 k_\phi}\right)} \, dk_\phi  \\&   \label{sigma_elle_p3} 
+ \int_{k_{\phi,2}}^{k_{\phi,3}} \sqrt{1 - \min\left(1, \frac{1}{a_{21}^2 \cos^2 k_\phi}\right)} \, dk_\phi  - \int_{k_{\phi,2}}^{k_{\phi,3}} \sqrt{1 - \min\left(1, \frac{1}{a_{22}^2 \cos^2 k_\phi}\right)} \, dk_\phi  \\&   \notag
+ \int_{k_{\phi,3}}^{k_{\phi,4}} \sqrt{1 - \min\left(1, \frac{1}{a_{31}^2 \cos^2 k_\phi}\right)} \, dk_\phi  - \int_{k_{\phi,3}}^{k_{\phi,4}} \sqrt{1 - \min\left(1, \frac{1}{a_{32}^2 \sin^2 k_\phi}\right)} \, dk_\phi \Bigg), \quad \ell\ge m
\end{align}
where ${a_{11} = \frac{L_y}{\lambda m}}$, ${a_{12} = \frac{L_x}{\lambda (\ell+1)}}$, ${a_{21} = \frac{L_x}{\lambda \ell}}$, ${a_{22} = \frac{L_x}{\lambda (\ell+1)}}$, ${a_{31} =\frac{L_x}{\lambda \ell}}$, and ${a_{32} = \frac{L_y}{\lambda (m+1)}}$.
The above expression can be rewritten equivalently as
\begin{align} \notag
&  \sigma_{\ell,m}^2   =  \frac{1}{4\pi} \Bigg(\int_{\max\left(k_{\phi,1},\mathrm{arcsin}\left(\frac{1}{a_{11}}\right)\right)}^{\max\left(k_{\phi,2},\mathrm{arcsin}\left(\frac{1}{a_{11}}\right)\right)} \sqrt{1 - \frac{1}{a_{11}^2 \sin^2 k_\phi}} \, dk_\phi  - \int_{\min\left(k_{\phi,1},\mathrm{arccos}\left(\frac{1}{a_{12}}\right)\right)}^{\min\left(k_{\phi,2},\mathrm{arccos}\left(\frac{1}{a_{12}}\right)\right)} \sqrt{1 - \frac{1}{a_{12}^2 \cos^2 k_\phi}} \, dk_\phi  \\& \notag \!\!\!\!\!
+ \int_{\min\left(k_{\phi,2},\mathrm{arccos}\left(\frac{1}{a_{21}}\right)\right)}^{\min\left(k_{\phi,3},\mathrm{arccos}\left(\frac{1}{a_{21}}\right)\right)} \sqrt{1 - \frac{1}{a_{21}^2 \cos^2 k_\phi}} \, dk_\phi  - \int_{\min\left(k_{\phi,2},\mathrm{arccos}\left(\frac{1}{a_{22}}\right)\right)}^{\min\left(k_{\phi,3},\mathrm{arccos}(\frac{1}{a_{22}})\right)} \sqrt{1 - \frac{1}{a_{22}^2 \cos^2 k_\phi}} \, dk_\phi  \\&   \label{sigma_elle_p4} \!\!\!\!\!
+ \int_{\min\left(k_{\phi,3},\mathrm{arccos}\left(\frac{1}{a_{31}}\right)\right)}^{\min\left(k_{\phi,4},\mathrm{arccos}\left(\frac{1}{a_{31}}\right)\right)} \sqrt{1 - \frac{1}{a_{31}^2 \cos^2 k_\phi}} \, dk_\phi  - \int_{\max\left(k_{\phi,3},\mathrm{arcsin}\left(\frac{1}{a_{32}}\right)\right)}^{\max\left(k_{\phi,4},\mathrm{arcsin}\left(\frac{1}{a_{32}}\right)\right)} \sqrt{1 - \frac{1}{a_{32}^2 \sin^2 k_\phi}} \, dk_\phi\Bigg)
\end{align}
where we used the fact that each of the above integrals yields non-zero value in the intervals $k_\phi \in [\arcsin(1/a), \pi/2]$ $\left(k_\phi \in [0, \mathrm{arccos}(1/a)]\right)$, where $a>1$ is one of the parameters listed above. 
The variances $\sigma_{\ell,m}^2$ are obtained by solving the following indefinite integrals \cite[Ch.~2.611]{IntegralBook}
\begin{align}
\int \sqrt{1 - \frac{1}{a^2 \sin^2 k_\phi}} \, dk_\phi & = \frac{1}{a} \arctan\left(\frac{\cos k_\phi}{\sqrt{a^2 \sin^2 k_\phi - 1}} \right) - a \arcsin \left(\frac{\cos k_\phi}{ \sqrt{1-1/a^2}} \right) \\
\int \sqrt{1 - \frac{1}{a^2 \cos^2 k_\phi}} \, dk_\phi & = -\frac{1}{a}  \arctan\left(\frac{\sin k_\phi}{\sqrt{a^2 \cos^2 k_\phi - 1}} \right) + \arcsin \left( \frac{\sin k_\phi}{ \sqrt{1-1/a^2}} \right).
\end{align}
Mathematical details are omitted for space limitations and will be provided upon request.
Finally, the case ${\ell< m}$ can be treated similarly by exchanging $k_{\phi,2}$ and $k_{\phi,3}$ so that now $k_{\phi,2} = \arctan\left(\frac{(m+1) L_x}{(\ell+1) L_y}\right)$ and $k_{\phi,3} = \arctan\left(\frac{m L_x}{\ell L_y}\right)$, and change the integration limits of \eqref{sigma_elle_p2} accordingly, 
which can be again solved by following the same procedure explained above.

\bibliographystyle{IEEEtran}
\bibliography{IEEEabrv,JSAC_SmallScale}

\begin{thebibliography}{10}
\providecommand{\url}[1]{#1}
\csname url@samestyle\endcsname
\providecommand{\newblock}{\relax}
\providecommand{\bibinfo}[2]{#2}
\providecommand{\BIBentrySTDinterwordspacing}{\spaceskip=0pt\relax}
\providecommand{\BIBentryALTinterwordstretchfactor}{4}
\providecommand{\BIBentryALTinterwordspacing}{\spaceskip=\fontdimen2\font plus
\BIBentryALTinterwordstretchfactor\fontdimen3\font minus
  \fontdimen4\font\relax}
\providecommand{\BIBforeignlanguage}[2]{{%
\expandafter\ifx\csname l@#1\endcsname\relax
\typeout{** WARNING: IEEEtran.bst: No hyphenation pattern has been}%
\typeout{** loaded for the language `#1'. Using the pattern for}%
\typeout{** the default language instead.}%
\else
\language=\csname l@#1\endcsname
\fi
#2}}
\providecommand{\BIBdecl}{\relax}
\BIBdecl

\bibitem{Marzetta2018}
T.~Marzetta, ``Spatially-stationary propagating random field model for {Massive
  MIMO} small-scale fading,'' in \emph{IEEE International Symposium on
  Information Theory, 2018. Proceedings.}, June 2018.

\bibitem{Marzetta2010}
T.~L. Marzetta, ``Noncooperative cellular wireless with unlimited numbers of
  base station antennas,'' \emph{IEEE Trans. Wireless Commun.}, vol.~9, no.~11,
  pp. 3590--3600, November 2010.

\bibitem{MarzettaBook}
T.~L. Marzetta, E.~G. Larsson, H.~Yang, and H.~Q. Ngo, \emph{Fundamentals of
  Massive {MIMO}}.\hskip 1em plus 0.5em minus 0.4em\relax Cambridge University
  Press, 2016.

\bibitem{LucaBook}
E.~Bj{\"{o}}rnson, J.~Hoydis, and L.~Sanguinetti, \emph{Massive {MIMO}
  Networks: Spectral, Energy, and Hardware Efficiency}.\hskip 1em plus 0.5em
  minus 0.4em\relax Foundations and Trends in Signal Processing, 2017, vol.~11,
  no. 3-4.

\bibitem{Five}
E.~{Bj{\"o}rnson}, L.~{Sanguinetti}, H.~{Wymeersch}, J.~{Hoydis}, and T.~L.
  {Marzetta}, ``{Massive MIMO is a Reality – What is Next? Five Promising
  Research Directions for Antenna Arrays},'' \emph{Digital Signal Processing},
  vol.~94, no.~3, pp. 3--20, November 2019.

\bibitem{Ngo2013}
H.~Q. {Ngo}, E.~G. {Larsson}, and T.~L. {Marzetta}, ``Energy and spectral
  efficiency of very large multiuser {MIMO} systems,'' \emph{IEEE Trans.
  Commun.}, vol.~61, no.~4, pp. 1436--1449, April 2013.

\bibitem{UnlimitedCapacity}
E.~Bj{\"{o}}rnson, J.~Hoydis, and L.~Sanguinetti, ``Massive {MIMO} has
  unlimited capacity,'' \emph{IEEE Trans. Wireless Commun.}, vol.~17, no.~1,
  pp. 574--590, Jan 2018.

\bibitem{Sanguinetti_2015}
E.~{Bj\"ornson}, L.~{Sanguinetti}, J.~{Hoydis}, and M.~{Debbah}, ``Optimal
  design of energy-efficient multi-user {MIMO} systems: {Is Massive MIMO} the
  answer?'' \emph{IEEE Trans. Wireless Commun.}, vol.~14, no.~6, pp.
  3059--3075, June 2015.

\bibitem{Cheng2017}
H.~V. {Cheng}, E.~{Bj\"ornson}, and E.~G. {Larsson}, ``Optimal pilot and
  payload power control in single-cell {Massive MIMO} systems,'' \emph{IEEE
  Trans. Signal Process.}, vol.~65, no.~9, pp. 2363--2378, May 2017.

\bibitem{Sanguinetti2019a}
L.~Sanguinetti, E.~Bj\"ornson, and J.~Hoydis, ``Towards massive {MIMO} 2.0:
  Understanding spatial correlation, interference suppression, and pilot
  contamination,'' \emph{IEEE Transactions on Communications}, vol.~68, no.~1,
  pp. 232--257, January 2020.

\bibitem{Prather2017}
D.~W. {Prather}, S.~{Shi}, G.~J. {Schneider}, P.~{Yao}, C.~{Schuetz},
  J.~{Murakowski}, J.~C. {Deroba}, F.~{Wang}, M.~R. {Konkol}, and D.~D. {Ross},
  ``Optically upconverted, spatially coherent phased-array-antenna feed
  networks for beam-space {MIMO} in {5G} cellular communications,'' \emph{IEEE
  Transactions on Antennas and Propagation}, vol.~65, no.~12, pp. 6432--6443,
  Dec 2017.

\bibitem{Rusek2018}
S.~Hu, F.~Rusek, and O.~Edfors, ``Beyond {Massive MIMO}: The potential of data
  transmission with large intelligent surfaces,'' \emph{IEEE Trans. Signal
  Process.}, vol.~66, no.~10, pp. 2746--2758, May 2018.

\bibitem{Pivotal}
\BIBentryALTinterwordspacing
E.~J. Black, ``Holographic beam forming and {MIMO},'' Pivotal Commware, Tech.
  Rep., 12 2017. [Online]. Available:
  \url{https://pivotalcommware.com/wp-content/uploads/2017/12/Holographic-Beamforming-WP-v.6C-FINAL.pdf}
\BIBentrySTDinterwordspacing

\bibitem{TseBook}
D.~Tse and P.~Viswanath, \emph{Fundamentals of Wireless Communication}.\hskip
  1em plus 0.5em minus 0.4em\relax Cambridge University Press, 2005.

\bibitem{Dardari2019a}
D.~Dardari, ``Communication with large intelligent surfaces: {F}undamental
  limits and models,'' \emph{CoRR}, vol. abs/1912.01719, 2019.

\bibitem{Poon2004}
A.~S.~Y. {Poon}, R.~W. {Brodersen}, and D.~N.~C. {Tse}, ``A spatial channel
  model for multiple-antenna systems,'' in \emph{IEEE Antennas and Propagation
  Society Symposium, 2004.}, vol.~4, June 2004, pp. 3665--3668 Vol.4.

\bibitem{Poon2006}
A.~S.~Y. {Poon}, D.~N.~C. {Tse}, and R.~W. {Brodersen}, ``Impact of scattering
  on the capacity, diversity, and propagation range of multiple-antenna
  channels,'' \emph{IEEE Trans. Inf. Theory}, vol.~52, no.~3, pp. 1087--1100,
  March 2006.

\bibitem{Clarke}
R.~H. {Clarke}, ``A statistical theory of mobile-radio reception,'' \emph{The
  Bell System Technical Journal}, vol.~47, no.~6, pp. 957--1000, July 1968.

\bibitem{Aulin79}
T.~{Aulin}, ``A modified model for the fading signal at a mobile radio
  channel,'' \emph{IEEE Trans. Veh. Technol.}, vol.~28, no.~3, pp. 182--203,
  Aug 1979.

\bibitem{Aris2000}
A.~L. Moustakas, H.~U. Baranger, L.~Balents, A.~M. Sengupta, and S.~H. Simon,
  ``Communication through a diffusive medium: Coherence and capacity,''
  \emph{Science}, vol. 287, no. 5451, pp. 287--290, 2000.

\bibitem{PaulrajBook}
A.~Paulraj, R.~Nabar, and D.~Gore, \emph{{Introduction to Space-Time Wireless
  Communications}}.\hskip 1em plus 0.5em minus 0.4em\relax Cambridge, UK:
  Cambridge University Press, 2003.

\bibitem{MolischBook}
A.~F. Molisch, \emph{Wireless Communications, Second Edition}.\hskip 1em plus
  0.5em minus 0.4em\relax Wiley, 2010.

\bibitem{Baggeroer}
\BIBentryALTinterwordspacing
A.~B. Baggeroer, ``Space/time random processes and optimum array processing,''
  Naval Undersea Center, San Diego (CA), Tech. Rep. ADA035593, Apr. 1976.
  [Online]. Available:
  \url{https://apps.dtic.mil/dtic/tr/fulltext/u2/a035593.pdf}
\BIBentrySTDinterwordspacing

\bibitem{Debbah2015}
Q.~{Nadeem}, A.~{Kammoun}, M.~{Debbah}, and M.~{Alouini}, ``A generalized
  spatial correlation model for {3D MIMO} channels based on the fourier
  coefficients of power spectrums,'' \emph{IEEE Trans. Signal Process.},
  vol.~63, no.~14, pp. 3671--3686, July 2015.

\bibitem{MarzettaIT}
T.~L. Marzetta, E.~G. Larsson, and T.~B. Hansen, \emph{Massive MIMO and
  Beyond}.\hskip 1em plus 0.5em minus 0.4em\relax Cambridge University Press,
  2019.

\bibitem{ChewBook}
W.~C. Chew, \emph{Waves and Fields in Inhomogenous Media}.\hskip 1em plus 0.5em
  minus 0.4em\relax Wiley-IEEE Press, 1995.

\bibitem{HildebrandBook}
F.~B. Hildebrand, \emph{{Advanced Calculus for Applications}}.\hskip 1em plus
  0.5em minus 0.4em\relax Prentice Hall, 1962.

\bibitem{StrattonBook}
A.~Stratton, \emph{Principles of Digital Communication}.\hskip 1em plus 0.5em
  minus 0.4em\relax Wiley, 2015.

\bibitem{Salazar}
T.~K. {Sarkar}, {Zhong Ji}, {Kyungjung Kim}, A.~{Medouri}, and
  M.~{Salazar-Palma}, ``A survey of various propagation models for mobile
  communication,'' \emph{IEEE Antennas and Propagation Magazine}, vol.~45,
  no.~3, pp. 51--82, June 2003.

\bibitem{FranceschettiBook}
M.~Franceschetti, \emph{Wave Theory of Information}.\hskip 1em plus 0.5em minus
  0.4em\relax Cambridge University Press, 2017.

\bibitem{Johnson}
\BIBentryALTinterwordspacing
S.~G. Johnson, ``When functions have no value(s): Delta functions and
  distributions,'' March 2017. [Online]. Available:
  \url{https://math.mit.edu/~stevenj/18.303/delta-notes.pdf}
\BIBentrySTDinterwordspacing

\bibitem{GallagerBook}
R.~G. Gallager, \emph{Principles of Digital Communication}.\hskip 1em plus
  0.5em minus 0.4em\relax Cambridge University Press, 2008.

\bibitem{Marzetta99}
T.~L. {Marzetta} and B.~M. {Hochwald}, ``Capacity of a mobile multiple-antenna
  communication link in rayleigh flat fading,'' \emph{IEEE Trans. Inf. Theory},
  vol.~45, no.~1, pp. 139--157, Jan 1999.

\bibitem{ArfkenBook}
G.~B. Arfken, H.~J. Weber, and F.~E. Harris, \emph{Mathematical Methods for
  Physicists}, G.~B. Arfken, H.~J. Weber, and F.~E. Harris, Eds.\hskip 1em plus
  0.5em minus 0.4em\relax Boston: Academic Press, 2013.

\bibitem{VanTreesBook}
H.~L. Van~Trees, \emph{Detection Estimation and Modulation Theory, Part
  I}.\hskip 1em plus 0.5em minus 0.4em\relax Wiley, 1968.

\bibitem{FFT99}
M.~Frigo, ``A fast fourier transform compiler,'' in \emph{Proceedings of the
  ACM SIGPLAN 1999 Conference on Programming Language Design and
  Implementation}, ser. PLDI '99.\hskip 1em plus 0.5em minus 0.4em\relax New
  York, NY, USA: ACM, 1999, pp. 169--180.

\bibitem{Pizzo2020}
A.~{Pizzo}, T.~L. {Marzetta}, and L.~Sanguinetti, ``Degrees of freedom of
  {Holographic MIMO} channels,'' in \emph{21st IEEE International Workshop on
  Signal Processing Advances in Wireless Communications}, Submitted to 2020.

\bibitem{Franceschetti}
M.~Franceschetti, ``On {L}andau's eigenvalue theorem and information
  cut-sets,'' \emph{IEEE Trans. Inf. Theory}, vol.~61, no.~9, pp. 5042--5051,
  Sept 2015.

\bibitem{PoonDoF}
A.~S.~Y. Poon, R.~W. Brodersen, and D.~N.~C. Tse, ``Degrees of freedom in
  multiple-antenna channels: a signal space approach,'' \emph{IEEE Trans. Inf.
  Theory}, vol.~51, no.~2, pp. 523--536, Feb 2005.

\bibitem{mmWave}
M.~Zhang, M.~Polese, M.~Mezzavilla, S.~Rangan, and M.~Zorzi, ``{Ns-3
  Implementation of the 3GPP MIMO Channel Model for Frequency Spectrum above 6
  GHz},'' in \emph{Proceedings of the Workshop on Ns-3}, ser. WNS3 ’17.\hskip
  1em plus 0.5em minus 0.4em\relax New York, NY, USA: Association for Computing
  Machinery, 2017, p. 71–78.

\bibitem{Thz}
K.~Sengupta, T.~Nagatsuma, and D.~Mittleman, ``Terahertz integrated electronic
  and hybrid electronic-photonic systems,'' \emph{Nature Electronics}, vol.~1,
  no.~12, pp. 622--635, Dec 2018.

\bibitem{AbramowitzStegun}
M.~Abramowitz and I.~Stegun, \emph{Handbook of mathematical functions}.\hskip
  1em plus 0.5em minus 0.4em\relax New York, NY: Dover Publications, 1972.

\bibitem{IntegralBook}
I.~S. Gradshteyn and I.~M. Ryzhik, \emph{{Table of Integrals, Series, and
  Products}}.\hskip 1em plus 0.5em minus 0.4em\relax Academic Press, 1965.

\end{thebibliography}
\end{document}